\documentclass{article}

\usepackage{fullpage}
\usepackage{graphicx}
\graphicspath{{figs/}}
\usepackage{amsmath,amssymb,amsfonts, amsthm}
\usepackage{xcolor}
\usepackage{xspace}
\usepackage{thm-restate}

\usepackage[colorlinks,linkcolor=blue,filecolor=blue,citecolor=blue,urlcolor=blue,pagebackref]{hyperref}

\usepackage{soul}

\usepackage{algorithm}
\usepackage{algpseudocode}
\algtext*{EndWhile}
\algtext*{EndIf}
\algtext*{EndFor}
\algrenewcommand\algorithmicrequire{\textbf{Input:}}
\algrenewcommand\algorithmicensure{\textbf{Output:}}

\usepackage[colorinlistoftodos,textsize=tiny,textwidth=2cm,color=red!25!white,obeyFinal]{todonotes}

\usepackage{enumitem}
\setlist{topsep=3pt, itemsep=0pt}
\newcommand{\R}{\mathbb{R}}
\newcommand{\cost}{\mathrm{cost}}

\newtheorem{theorem}{Theorem}
\newtheorem{lemma}[theorem]{Lemma}
\newtheorem{definition}[theorem]{Definition}
\newtheorem{claim}[theorem]{Claim}

\newtheorem{remark}{Remark}

\DeclareMathOperator*{\E}{{\mathbb{E}}}
\DeclareMathOperator*{\union}{\bigcup}

\newcommand{\calA}{{\mathcal{A}}}
\newcommand{\calC}{{\mathcal{C}}}
\newcommand{\calK}{{\mathcal{K}}}
\newcommand{\calP}{{\mathcal{P}}}
\newcommand{\calQ}{{\mathcal{Q}}}

\newcommand{\poly}{{\mathrm{poly}}}

\newcommand{\tG}{\tilde{G}}

\newcommand{\pedge}{{$+$edge}\xspace}
\newcommand{\medge}{{$-$edge}\xspace}
\newcommand{\pedges}{{$+$edges}\xspace}
\newcommand{\medges}{{$-$edges}\xspace}

\newcommand{\prepinstance}{(\calK, E_{\adm})}

\newcommand{\adm}{\mathrm{adm}}
\newcommand{\err}{\mathrm{err}}

\newcommand{\opt}{{\mathrm{opt}}}

\newcommand{\epsrt}{\varepsilon_{r}}
\newcommand{\cc}{{Correlation Clustering}\xspace}

\newcommand{\costr}{\mathrm{cost}^r}
\newcommand{\costi}{\mathrm{cost}^i}

\newcommand{\lpr}{\Delta^r}
\newcommand{\lpi}{\Delta^i}

\newcommand{\eps}{\varepsilon}

\newcommand{\epsfour}{\eps}

\newcommand{\epsq}{\eps_q}

\newcommand{\epsa}{\eps_a}

\begin{document}

\title{Handling Correlated Rounding Error via Preclustering:\\ A 1.73-approximation for Correlation Clustering}

\author{ 
{Vincent Cohen-Addad} \\ Google Research\\ \texttt{cohenaddad@google.com}
\and
{Euiwoong Lee\footnote{Supported in part by NSF grant CCF-2236669 and Google.}} \\ University of Michigan\\ \texttt{euiwoong@umich.edu}
\and {Shi Li}\\ Nanjing University\\ \texttt{shili@nju.edu.cn} 
\and 
{Alantha Newman} \\ Université Grenoble Alpes\\ \texttt{alantha.newman@grenoble-inp.fr} 
}

\date{}

\maketitle

\begin{abstract}
    We consider the classic Correlation Clustering problem: Given a complete graph where edges are labelled either $+$ or $-$, the goal is to find a partition of the vertices that minimizes the number of the \pedges across parts plus the number of the \medges within parts.  Recently, Cohen-Addad, Lee and Newman~\cite{CLN22} presented a 1.994-approximation algorithm for the problem using the Sherali-Adams hierarchy, hence breaking through the integrality gap of 2 for the classic linear program and improving upon the 2.06-approximation of Chawla, Makarychev, Schramm and Yaroslavtsev~\cite{CMSY15}. 

We significantly improve the state-of-the-art by providing a 1.73-approximation for the problem. Our approach introduces a preclustering of Correlation Clustering instances that allows us to essentially ignore the error arising from the {\em correlated rounding} used by \cite{CLN22}.
This additional power simplifies the previous algorithm and analysis.  
More importantly, it enables a 
 new {\em set-based rounding} that complements the previous roundings.
A combination of these two rounding algorithms yields the improved bound.
\end{abstract}

\thispagestyle{empty}
\newpage
\thispagestyle{empty}
\newpage
\pagenumbering{arabic} 
    \section{Introduction}
Clustering is a classic problem in unsupervised machine learning and
data mining. Given a set of data elements and pairwise similarity
information between the elements, the goal of clustering is to find a
partition of the data elements such that elements in the same clusters
are pairwise similar while elements in different clusters are pairwise
dissimilar.  Introduced by Bansal, Blum and Chawla~\cite{BBC04},
Correlation Clustering has become one of the most widely studied
formulations for clustering. The input of the problem consists of a
complete graph $(V, E^+ \uplus E^-)$, where $E^+ \uplus E^- = {V
  \choose 2}$, $E^+$ representing the so-called \emph{positive} edges
and $E^-$ the so-called \emph{negative} edges. The goal is to find a
partition of the vertex set so as to minimize the number of
\emph{unsatisfied} edges, namely the number of the negative edges $uv$
where $u$ and $v$ are in the same cluster plus the number of the
positive edges $uv$ where $u$ and $v$ are in different clusters.
Here, the vertex set represents the elements to cluster while positive
edges represent pairs of similar elements and negative edges pairs of
dissimilar elements.  The above formulation is very basic and has thus
led the Correlation Clustering problem to encompass a variety of
applications from finding clustering ensembles
\cite{bonchi2013overlapping}, duplicate detection
\cite{arasu2009large}, community mining \cite{chen2012clustering},
disambiguation tasks \cite{kalashnikov2008web}, to automated labelling
\cite{agrawal2009generating, chakrabarti2008graph} and many more.

The problem is known to be NP-hard, and so the focus has been on
designing approximation algorithms for the problem.  In their seminal
paper, Bansal, Blum and Chawla~\cite{BBC04} gave an
$O(1)$-approximation algorithm for the problem, which was later
improved by Charikar, Guruswami and Wirth~\cite{CGW05} to a
4-approximation, obtained by rounding the natural linear program (LP)
relaxation for the problem. Charikar, Guruswami and Wirth~\cite{CGW05}
also showed that the problem is APX-Hard.  Soon after, Ailon, Charikar
and Newman~\cite{ACN08} introduced an influential {\em pivot-based}
algorithm, that leads to a combinatorial $3$-approximation and a
LP-based $2.5$-approximation. The pivot-based algorithms continue to
inspire new results for Correlation Clustering in different settings
(see Section~\ref{sec:furtherrelatedwork} for more details).  Ten
years after the results of \cite{ACN08}, Chawla, Makarychev, Schramm
and Yaroslavtsev~\cite{CMSY15} further improved the LP rounding scheme
and obtained a $2.06$-approximation, nearly matching the LP
integrality gap of 2. Since the best known approximation results have
been obtained through LP rounding techniques and the LP has an
integrality gap of 2, this bound seemed to be an important roadblock
in the direction of getting better approximation bounds. Recently,
Cohen-Addad, Lee and Newman~\cite{CLN22} broke through this barrier
using $O(1/\eps^2)$ rounds of the Sherali-Adams hierarchy on top of
the standard LP and rounding the resulting fractional solution to
obtain a $(1.994+\eps)$-approximation algorithm.

The result of \cite{CLN22} shows the importance of hierarchies to
decrease the approximation ratio for the problem. In fact, for a large
constant number of rounds of the Sherali-Adams hierarchy, we do not
know a lower bound on the integrality gap of the resulting LP, which
may even perhaps yield optimal results under P$\neq$NP or the Unique
Games Conjecture. Presumably, the approximation ratio of
$(1.994+\eps)$ obtained through rounding the Sherali-Adams relaxation
does not reflect the actual power of the hierarchies, but rather the
limitations of the current rounding approaches and of their analysis.
Thus, to improve over the $(1.994+\eps)$-approximation one might have
to go beyond the pivot-based rounding framework introduced by
\cite{ACN08}, further developed by \cite{CMSY15} and \cite{CLN22},
which has been the only known way of obtaining a better than
3-approximation for the problem for the last 15 years.  This is what
we propose in this paper.

\subsection{Our Results}
We present a drastic improvement over the result of Cohen-Addad, Lee
and Newman~\cite{CLN22} by showing a $(1.73+\eps)$-approximation to
the problem using the Sherali-Adams hierarchy.

\begin{restatable}{theorem}{main}{\rm (Main Result)}\label{thm:main}
For any $\eps > 0$, there exists a $(1.73+\eps)$-approximation
algorithm for Correlation Clustering with running time $n^{O(1/\poly(\eps))}$.
\end{restatable}

In addition to the above bound, our contributions are the
following. Our new result departs from previous work in several ways:
\begin{enumerate}
    \item We provide a preclustering step that identifies pairs of
      vertices that ``clearly'' belong to the same cluster/do not
      belong to the same cluster in an optimum solution. This is
      achieved by computing a group of vertices whose contribution to
      the objective function in an optimum solution is tiny compared
      to the number of incident \pedges, and results in a
      preclustering where the number of ``undecided'' pairs can be
      upper bounded in terms of the cost of the optimal clustering; we
      hence get a useful lower bound on the cost contribution of the
      remaining instance.
    
    \item Equipped with this, we provide a new {\em set-based
      rounding} approach inspired by the work of Kleinberg and
      Tardos~\cite{KT02} for the Uniform Metric Labeling problem. To
      the best of our knowledge, it is the first application of the
      ideas from~\cite{KT02} to Correlation Clustering.  The original
      work~\cite{KT02}, when adapted to Correlation Clustering,
      creates a variable $y_S$ for each $S \subseteq V$ that can
      possibly become a cluster and sample a set with probability
      proportional to $y_S$'s (hence the name {\em set-based}).
    
Of course, it is impossible to create a variable for each subset, and
we therefore need to use the power of the Sherali-Adams hierarchy and
the {\em correlated rounding} technique introduced by Raghavendra and
Tan~\cite{RT12} to extend this approach to output arbitrary-size
subsets when needed. This comes at a price: there is an additive
rounding error proportional to the total number of vertices squared,
which is prohibitive when the optimal value is $o(n^2)$.  However,
this is where the preclustering step comes to the rescue; we use the
correlated rounding only for the ``undecided pairs'' whose number is
small compared to $\opt$, so that this additive rounding error is
negligible.
    
    \item This way of handling error from the correlated rounding also
      provides a refined analysis of the Cohen-Addad, Lee and
      Newman~\cite{CLN22} rounding, which combines the classical {\em
        pivot-based rounding} with the correlated rounding.  As their
      main technical complications arise from the rounding error, our
      preclustering step considerably simplifies their analysis.
      While following the same triangle-based analysis of
      \cite{CLN22}, we introduce a new edge-by-edge charging argument
      that helps this algorithm nicely complement the set-based
      rounding above.
    
    \item In the final step, we combine the two above rounding
      approaches to show that the bad cases for one type of edges
      ($E^+$ or $E^-$) are the good cases of the other; our combined
      rounding gives the desired bound of $1.73+\eps$.
\end{enumerate}
We next present a more detailed overview of the techniques.

\subsection{Our Techniques}
To understand in more details our contribution and the new techniques
introduced, we need to provide a brief summary of the previous
approaches. We first recall the classic LP relaxation (whose
integrality gap is known to be in $[2, 2.06]$). There is a variable
$x_{uv}$ for each pair of vertices whose intended value is 1 if $u,v$
are not in the same cluster and $0$ otherwise. The goal is thus to
minimize $\sum_{uv \in E^+} x_{uv} + \sum_{uv \in E^-} (1-x_{uv})$,
under the classic triangle inequality constraint: $\forall u,v,w$,
$x_{uv} \le x_{uw} + x_{wv}$.

To obtain the 2.5-approximation algorithm, Ailon, Charikar and
Newman~\cite{ACN08} designed a pivot-based rounding method that
proceeds as follows: (1) Pick a random vertex $p$, called the
\emph{pivot}, and create the cluster containing $p$; and (2) Recurse
on the rest of the instance (the vertices not in the cluster).  This
hence builds the clustering in a sequential manner. Once we commit to
this scheme, the main design question that remains is how to construct
the cluster containing the pivot.  The solution of \cite{ACN08} was to
go over all the other vertices and for each other vertex $u$ place it
in the cluster of the pivot $p$ with probability $(1-x_{pu})$,
independently of the other random decisions made for the other
vertices.

This was further improved by Chawla, Makarychev, Schramm and
Yaroslavtsev~\cite{CMSY15} who kept the same scheme but provided a
better rounding approach: For $+$edges, they replaced the probability
$(1-x_{pu})$ with a new rounding function $f^+$ to apply to the
quantity $(1-x_{pu})$ yielding the probability of incorporating $u$ in
$p$'s cluster.  Similar to the analysis of \cite{ACN08}, the analysis
of the rounding scheme was triangle-based: The analysis is a charging
scheme that charges the cost paid by each pair $u,v$ to the triangle
$p,u,v$ where $p$ is the pivot which decided edge $u,v$.  Then the
crux of the analysis is to show that for any triangle $p,u,v$, the
charge is bounded compared to the LP cost.

The approach of \cite{CLN22}, while using the same pivot-based
rounding and triangle-based analysis, introduces two twists.  The
first twist is the usage of \emph{correlated rounding} based on the
Sherali-Adams hierarchy. Namely, given a pivot $p$, the set of
$+$neighbors of $p$ that join the cluster of $p$ is chosen in a
correlated manner, using the techniques of Raghavendra and
Tan~\cite{RT12}.  Concretely, the Sherali-Adams hierarchy provides
variables of the form $y_{S}$ for any constant-sized set of vertices
$S$ that indicates the probability that all the vertices in $S$ are in
the same cluster.\footnote{In which case, $y_{uv} = 1-x_{uv}$.}  Given
a pivot $p$, the correlated rounding then allows to sample in such a
way that the probability that $u,v$ join the cluster of $p$ is
$y_{puv} \pm \eps$, where $\eps$ is an arbitrarily small constant.  An
important issue is that this additive $\eps$ error is in some cases
not a relative error (think of $(u,v) \in E^+$ and $x_{uv} =
0$). Dealing with these additive errors is one of the most technical
parts of the contribution of \cite{CLN22} and perhaps a limitation to
getting an improved approximation ratio (e.g., special attention must
be paid to edges with values close to $0$ and $1$ in both the
algorithm and in the analysis).  This first twist allows to bring the
approximation ratio down to 2 but not further; the ratio of some of
the triangles is 2. To bypass this bound, the second twist in the
approach is to use the properties of the Sherali-Adams hierarchy to
argue that one cannot pack too many bad triangles without creating
some good triangles (for which the ratio is smaller than 2), hence
bringing down the ratio to $(1.994+\eps)$ via a charging argument.  We
next describe how we go beyond previous work.

\paragraph{The High-Level Approach} The governing high-level 
intuition for what we are trying to achieve is this. Given a pivot
vertex $p$, the key property of the correlated rounding (see
Lemma~\ref{lem:RTrestated} for a restatement) ensures that we can
sample a set $C$ from a set of vertices $S$ both being of arbitrary
size (importantly not necessarily constant) such that pairs $u,v$ ends
up in $S$ with probability $y_{puv}$ on average, up to losing an
additive $\eps |S|^2$ error in the cost. This is one of the key
properties of the Sherali-Adams hierarchy that can be used to obtain
approximation schemes for dense CSPs, in particular Dense
Max-Cut~\cite{de2007linear,yoshida2014approximation}.  In light of the
previous approaches, this is a very desirable property: We could
repeatedly construct a cluster using this tool and the error would
only be coming from the fact that the clusters are constructed
sequentially.

The main issue we need to face with the above plan of attack is that
the additive error of $\eps|S|^2$, which we pay when we create a
cluster $C$, may end up being much larger than the cost of the optimum
solution. In fact, an optimum solution may have cost zero.  So our
goal is the following: (1) precluster the instance so that the cost of
the optimal solution becomes relatively large compared to the size of
the ``unsure part'', and (2) use correlated rounding only for the
unsure part of the instance so that the error from the correlated
rounding is negligible.

\paragraph{A New Preclustering} Our first conceptual contribution is a preclustering 
of the instance that aims at identifying ``clear clusters''. More
concretely, our goal is to make sure that if we apply the correlated
rounding to create a cluster of size $s$, then we can accommodate an
additive cost of $\eps s^2$. Thus, consider the clusters of the
optimum (integral) solution; when is the cost contribution of such a
cluster $S$ significantly smaller than $s^2$, say $\eps s^2$?  For
this to happen, it must be that the number of \medges internal to $S$
is small compared to ${s \choose 2}$ and that the number \pedges with
exactly one endpoint in $S$ is also much smaller than ${s \choose
  2}$. This cluster is thus almost a clique with tiny
\pedge-expansion; we formalize this through the notion of an
\emph{atom} in Section~\ref{sec:preprocess}.

This means that for a very large fraction of the pairs $u,v$ of
vertices in an atomic cluster $S$, the $+$neighborhood of $u$ and $v$
are nearly identical (up to a tiny fraction).  This is exactly what
the notion of \emph{agreement} introduced by Cohen-Addad, Lattanzi,
Mitrovic, Norouzi-Fard, Parotsidis and Tarnawski~\cite{CLMNP21} (see
Definition~\ref{definition:agreement}) was designed for.  Indeed, in
this work, the authors presented a massively-parallel constant factor
approximation algorithm. The constant obtained there is larger than
500, but the key idea is that it is enough to identify the atoms,
since all the other clusters $C$ pay a cost of at least $\eps \sum_{u
  \in C} d_u$, where $d_u$ is the number of $+$neighbors of $u$, and
therefore making all the remaining vertices singleton clusters would
be enough to get an $O(1/\eps)$-approximation.

We thus re-use the Algorithm 1 of~\cite{CLMNP21} and show that it
correctly identifies the atoms.  \cite{CLMNP21} then showed that there
exists a $(1+\eps)$-approximate solution $\calC^*$, such that each
non-singleton cluster output by the Algorithm 1 of~\cite{CLMNP21}, now
referred to as an atom, is fully contained within one cluster of
$\calC^*$. This immediately allows us to treat each atom as a single
(weighted) vertex, or in other words enforce $x_{uv} = 0$ for all
pairs $u,v$ in the same atom.

The next key step is that for each vertex $u$, we can identify a set
of so-called \emph{admissible} edges, which induce a neighborhood
$N'_u$ of size at most $d_u / \eps^{O(1)}$.  We then show that there
exists a near-optimum solution such that if $u,v$ is not an atomic or
an admissible pair, then $u$ and $v$ are not in the same cluster. For
such pairs $u,v$, we can immediately set $x_{uv} = 1$ in the LP,
because we know in advance that they must be placed in separate
clusters. This immediately reduces the uncertainty of the fractional
LP solution: for any vertex $u$, the number of pairs $u,v$ for which
the LP can have a fractional value (i.e., a value not in $\{0,1\}$) is
at most the number of admissible pairs which we show through a
delicate argument is at most $\sum_{u \in V} d_u/\eps^{O(1)}$.

The final step of the preclustering is then to show that the cost of
an integral solution under the above constraints is at least
$\eps^{O(1)}$ times the number of admissible pairs.  Combined with the
above argument, we can show that this enables the use of correlated
rounding while only losing an overall $(1+\eps)$ multiplicative
approximation, which is crucial because both of our rounding
algorithms below use the correlated rounding.

We believe that this preclustering will be of particular interest to
future work on improving the approximation ratio given that it
provides a much more structured instance and it remains completely
independent from the rounding process.

\paragraph{A New Set-Based Rounding} Our other major contribution is the following {\em set-based} rounding, whose performance complements the previous ones~\cite{ACN08, CMSY15, CLN22}.
Namely, give a solution $x \in [0, 1]^{\binom{V}{2}}$ to some LP relaxation, our goal is to obtain (as stated in 
Theorem~\ref{thm:KT}) a clustering $\calC$ with
    \begin{align*}
        \cost(\calC) \leq \sum_{vw \in E^+} \frac{2x_{vw}}{1+x_{vw}} \quad + \quad \sum_{vw \in E^-} \frac{1-x_{vw}}{1+x_{vw}}  \quad + \quad O(\eps) \cdot |E_\adm|,
    \end{align*}
where $\cost(\calC)$ is the cost of the clustering $\calC$. The set $E_{\adm}$ refers to the set of admissible edges defined in the preclustering.

To achieve this we use the ideas from the $2$-approximation algorithm for the Uniform Metric Labeling problem by Kleinberg and Tardos~\cite{KT02}. Consider a special case of the Correlation Clustering problem where we are promised that the optimum clustering has maximum cluster size $t = O(1)$. Then, in our LP, we could simply have a variable $z_S$ for every set $S$ of size at most $t$, indicating if $S$ is a cluster or not. For every $v \in V$, we have $\sum_{S \ni v}z_S = 1$. In each iteration of the rounding algorithm, we randomly choose a set $S$ with probability to $z_S/Z$, where $Z := \sum_S z_S$. So an edge $vw$ is decided with probability $\frac{1+x_{vw}}Z$. For an edge $vw \in E^+$, it incurs a cost with probability $\frac{2x_{vw}}Z$; for an edge $vw \in E^-$, the probability is $\frac{1-x_{vw}}Z$. This gives expected cost of $\frac{2x_{vw}}{1+x_{vw}}$ and  $\frac{1-x_{vw}}{1+x_{vw}}$ for $+$ and $-$edges $vw$ respectively, proving the inequality.

The most prominent challenge is thus to enable this approach to work
beyond the assumption that the optimum clusters have constant
size. This may seem impossible at first since if such a technique
existed in general graph partitioning, this would provide an
approximation bound that would yield a result believed to be unlikely
(e.g., a $2$-approximation for Minimum Bisection). Thus, we need to
exploit the particularities of the Correlation Clustering problem. To
achieve this, we first enrich the program with variables of the form
$y_{S}^s$, for all $s \in [n]$, and subset $S$ of $1/\eps^{O(1)}$
vertices. Here $y_{S}^s$ aims at representing the number of clusters
of size $s$ containing $S$ as a subset, and should be in $\{0,1\}$ in
an integral solution when $S \neq \emptyset$.  Then, $y_S$ now
represents the number of clusters (of various sizes) containing $S$ as
a subset, and should also be in $\{0,1\}$ in an integral solution. It
thus indicates whether $S$ is a subset of a cluster.

We next show that this is enough to achieve the above bound. To do so,
we provide the following rounding approach: We first sample a size $s$
based on the $y^s_{\emptyset}$ probabilities, then a pivot $a$ that
will go in a cluster of size $s$ based on the $y_a^s$ probabilities
(correcting for the fact that we are aiming to build a cluster of size
$s$), and then use the correlated rounding to complete the cluster.
We show, perhaps surprisingly, that this is enough to achieve the
bounds offered by the approach of~\cite{KT02}, at the price of an
extra additive error incurred by the correlated rounding, which, as we
show, results in an additive $O(\eps) \cdot |E_{\adm}|$ additive
error. This error is tolerable thanks to the preclustering step.

\paragraph{A Refined and Cleaner Analysis of the Pivot-Based Rounding}
This preclustering immediately makes the analysis of \cite{CLN22}
rounding much simpler. All the corner cases that required a lot of
work to handle the additive error can be removed: The concepts of
short and long edges (edges $uv$ for which $x_{uv}$ was close to 0 or
1) can be completely forgotten and the rounding can be simplified. In
fact, \cite{CLN22} proposed a simpler and better so-called ``ideal''
analysis of their rounding if the additive error $\eps$ can be assumed
to be 0 (which of course cannot be achieved with $o(n)$ rounds of
Sherali-Adams) and our preclustering allows to immediately recover
this ideal analysis, modulo a simple cleaning step to handle the
atoms. This yields a very simple 2-approximation rounding scheme.

We then refine the bound obtained by \cite{CLN22} to express 
it as a per-edge approximation. Concretely, we show that the output clustering $\calC$
satisfies the following guarantees (as stated in Theorem~\ref{thm:CLN}):
    \begin{align*}
      \cost(\calC) \leq \sum_{vw \in E^+} \min\{1.515 + x_{vw}, 2\} \cdot x_{vw} \quad + \quad 2\sum_{vw \in E^-}(1-x_{vw})  \quad + \quad O(\eps) \cdot |E_\adm|.
          \end{align*}

\paragraph{Combining Two Roundings}
The two above bounds are simply combined as follows: We round the LP
with the pivot-based rounding of \cite{CLN22} with some carefully
chosen probability and the new rounding approach with the remaining
probability; for intuition, note that $-$edges have a
$1$-approximation in the first rounding and $2$ in the second
rounding, while $+$edges with small $x$ value have a $2$-approximation
in the first rounding but close to a $1.5$-approximation in the
second. Therefore, by appropriately combining the two roundings, we
can show that every edge has a $1.73$-approximation in expectation.
The preclustering ensures that the additive $O(\eps) \cdot |E_\adm|$
terms are only an $O(\poly(\eps))$ fraction of the optimum cost and
can thus be ignored.

\paragraph{Round-or-Cut Framework} 

In order to perform the set-based rounding and pivot-based rounding
{\em simultaneously}, our overall algorithmic framework after the
preclustering uses the round-or-cut paradigm.  In a typical rounding
algorithm, we are given a metric $x \in [0, 1]^{|V| \choose 2}$ over
$V$ along with some auxiliary variables, obtained from solving an LP
relaxation. We then round $x$ into some integral clustering with a
small cost.  This structure works for the pivot based rounding
algorithm: we randomly generate a cluster $C$ according to the
algorithm, remove $C$, focus on the LP relaxation restricted to $V
\setminus C$ (i.e., we reduce the LP solution), and repeat until we
clustered all vertices in $V$.  However, the adaptive feature of our
set-based rounding algorithm, which solves an LP extending $x$ in
every iteration, makes this one-shot strategy hard to implement.

Instead, we design a rounding algorithm with violation detection (Theorem~\ref{thm:round-or-cut}). In the algorithm, we are given the core vector $x \in [0, 1]^{|V| \choose 2}$ only. The algorithm proceeds in iterations. In each iteration, we try extend $x$ to a vector $(x, y)$ with the auxiliary variables $y$, which depend on the remaining set $V'$ of vertices. If the extension is successful, then we can proceed with the iteration by constructing the cluster $C$ and removing $C$. Otherwise, we find a hyperplane that separates $x$ from the convex hull of all integral clusterings, and return it. 
The well-known property of the ellipsoid algorithm~\cite{GLS12} will ensure that we will eventually find a desired clustering.

\subsection{Further Related Work}
\label{sec:furtherrelatedwork}
The Correlation Clustering problem has also been studied in the
weighted case, where each pair of vertices has an associated weight
and unsatisfied edges contribute a cost proportional to their weight
to the objective. Unfortunately, this version of the problem is
equivalent to the Multicut problem in terms of approximation guarantee
and so an $O(\log n)$-approximation is
known~\cite{demaine2006correlation} and improving this bound
significantly would be a major breakthrough. Moreover, no
polynomial-time constant factor approximation algorithm exists
assuming the Unique Game Conjecture~\cite{CKKRS06}.

The maximization version of the problem, where the goal is to maximize
the number of satisfied edges, has also been studied.  A PTAS for the
problem was given by Bansal, Blum and Chawla~\cite{BBC04}, and a
.77-approximation for the weighted case was given by Charikar,
Guruswami and Wirth~\cite{CGW05} and
Swamy~\cite{swamy2004correlation}.

In the unweighted case, a PTAS exists when the number of clusters is a
fixed constant~\cite{giotis2006correlation,karpinski2009linear}. A lot
of work has also been devoted to the minimization version of
Correlation Clustering in other computation models:
online~\cite{mathieu2010online,NEURIPS2021_250dd568,CLMP22}, more
practical settings, and in particular distributed or
parallel~\cite{chierichetti2014correlation,pmlr-v37-ahn15,CLMNP21,DBLP:conf/nips/PanPORRJ15,DBLP:conf/wdag/CambusCMU21,DBLP:conf/icml/Veldt22,DBLP:conf/www/VeldtGW18},
differential-privacy~\cite{bun2021differentially,Daogao2022,CFLMNPT22}. Related
to our results is the work of Cohen-Addad, Lattanzi, Mitrovic,
Norouzi-Fard, Parotsidis and Tarnawski~\cite{CLMNP21} whose approach
is useful to our preprocessing step, see also the work of Assadi and
Wang~\cite{DBLP:conf/innovations/Assadi022}. The above works have been
improved by Behnezhad, Charikar, Ma, and
Tan~\cite{DBLP:conf/focs/BehnezhadCMT22,DBLP:conf/soda/BehnezhadCMT23}
by showing how to extend the combinatorial pivot approach to the
distributed or streaming settings.  Note that recently, connections
between metric embeddings into ultrametric have been established and
Correlation Clustering plays a central role in the current best known
approximation
algorithms~\cite{DBLP:conf/focs/Cohen-Addad0KPT21,CFLM22}. In fact,
our new bound (marginally) improves the constant factor approximation
of Cohen-Addad, Das, Kipouridis, Parotsidis and
Thorup~\cite{DBLP:conf/focs/Cohen-Addad0KPT21}.

    \section{Overall Framework}
In this section, we present our overall framework to achieve a $(1.73 + \eps)$-approximation algorithm. 
On the way, we will also introduce our main technical results for the preclustering (Theorem~\ref{thm:preprocessed-wrapper}),
the set-based rounding 
(Theorem~\ref{thm:KT}), and
the pivot-based rounding 
(Theorem~\ref{thm:CLN}), 
which will be proved in Section~\ref{sec:preprocess},~\ref{sec:KT}, and~\ref{sec:CLN} respectively. 
We begin with some definitions related to our preclustering step. 

\begin{definition}\label{definition:prepro}
Given a Correlation Clustering instance $(V, E^+ \cup E^-)$, a {\em preclustered instance} is defined by a pair $(\calK, E_{\adm})$, where 
$\calK$ is a family of disjoint subsets of $V$ (not necessarily a partition), and $E_{\adm} \subseteq \binom{V}{2}$ is a set of pairs such that for every $uv \in E_\adm$, at least one of $u$ and $v$ is not in $\union_{K \in \calK}{K}$.

Each set $K \in \calK$ is called an {\em atom}. We use $V_\calK:= \union_{K \in \calK}K$ to denote the set of all vertices in atoms. A pair $(u, v)$ between two vertices $u,v$ in a same $K \in \calK$ is called an {\em atomic edge}. A pair that is neither an atomic nor an admissible edge is called a non-admissible edge. 
\end{definition}

Therefore, in a preclustered instance, the set $V \choose 2$ is partitioned into atomic, admissible and non-admissible edges. By the definition of $E_\adm$, a pair $(u, v)$ between two different atoms is non-admissible.

\begin{definition}
Given a preclustered instance $(\calK, E_\adm)$ for some Correlation Clustering instance $(V, E^+ \cup E^-)$,  a partition $\calC$ of $V$ (also called a clustering) is called {\em good} with respect to $(\calK, E_{\adm})$ if 
\begin{itemize}
    \item $u$ and $v$ are in the same set (or cluster) in $\calC$ for an atomic edge $(u, v)$, and
    \item $u$ and $v$ are not in the same set (or cluster) in $\calC$ for a non-admissible edge $(u, v)$.
\end{itemize}
\end{definition}
That is, a good clustering can not break an atom, or join a non-admissible pair. As a result, two atoms can not be in the same cluster as the edges between them are non-admissible.  The main theorem we prove for our preclustering step in Section~\ref{sec:preprocess} is the following: 
\begin{restatable}{theorem}{preprocessdwrapper}
{\rm (Preclustering)}
\label{thm:preprocessed-wrapper}
For any sufficiently small $\delta > 0$, there exists a polynomial-time algorithm that, given a Correlation Clustering instance $(V, E^+ \cup E^-)$ with the optimal value $\opt_0$, produces a preclustered instance $(\calK, E_{\adm})$ such that 
\begin{itemize}
    \item there exists a good clustering whose cost is at most $(1 + \delta) \opt_0$, and 
    \item $|E_{\adm}| \leq O(\opt_0 / \delta^{12})$. 
\end{itemize}
\end{restatable}

We remark that, after the preclustering, the benchmark clustering we are comparing to is a good clustering, but our algorithm does not need to find a good clustering. For example, the clustering we construct might join a non-admissible pair, or two atoms.
But we guarantee that an atom will not be broken by the clustering. 

Given the desired error parameter $\eps_0 > 0$ for the overall approximation factor in Theorem~\ref{thm:main}, 
let 
$\delta := \eps_0 / 4$, $\eps := \Theta(\delta^{12} \eps_0) = \Theta(\eps_0^{13})$ and perform the preclustering with parameter $\delta$ (let $\opt$ be the cost of the best good clustering with respect to the preclustering). 
Then, producing a clustering $\calC$ whose cost (denoted as $\cost(\calC)$) is at most $1.73\opt + \eps |E_{\adm}|$ guarantees that 
\[
\cost(\calC) \leq 1.73\opt + 
\eps |E_{\adm}| 
\leq 1.73(1 + \delta)\opt + 
O(\eps / \delta^{12}) \opt_0 \leq (1.73+\eps_0) \opt_0.
\]

Therefore, given a preclustered instance $(V, E^+ \cup E^-)$ together
with $(\calK, E_{\adm})$, it suffices to compute a clustering $\calC$
with $\cost(\calC) \leq 1.73\opt + \eps |E_{\adm}|$.

Our algorithm proceeds via the round-or-cut framework. 
Every clustering corresponds to a $0/1$-valued metric over $V$, where the distance between $u$ and $v$ indicates if $u$ and $v$ are cut in the clustering or not. We let $\calP \subseteq [0,1]^{V \choose 2}$ denote the convex hull of the metrics for all good clusterings. The central piece of the algorithm for the proof of Theorem~\ref{thm:main} is a rounding algorithm $\calA$ with violation detection. Formally, we prove the following theorem:
\begin{theorem}
    \label{thm:round-or-cut}
    Given a metric $x \in [0, 1]^{V \choose 2}$ over $V$, for which any two vertices in a same atom have distance $0$, and the two end points of any non-admissible edge have distance $1$, and $\eps > 0$, there is an $n^{O(1/\eps^4)}$ time algorithm $\calA$ that outputs one of the following two things:  
    \begin{itemize}
        \item a clustering for the instance whose cost is at most $1.73 \cdot\cost(x) + \eps |E_{\adm}|$, 
        \item a hyperplane separating $x$ and $\calP$ (i.e., a vector $w \in \R^{V \choose 2}$ and $b \in \R$ such that $w^T x' \geq b$ for every $x' \in \calP$, but $w^Tx < b$). Sometimes we also simply call the $(w, b)$ pair separating $x$ and $\calP$ a separation plane for $x$. 
    \end{itemize}    
\end{theorem}

It is well known~\cite{GLS12} that once we have the algorithm $\calA$ in Theorem~\ref{thm:round-or-cut}, we can combine it with the ellipsoid method to find the desired approximate solution for the preclustered instance
in polynomial time, which proves Theorem~\ref{thm:main}, with a final running time of $n^{O(1/\eps^4)} = n^{O(1/\eps_0^{52})}$. By enumeration or binary search, we assume we are given the value of $\opt$, and our goal is to find a clustering with cost at most $1.73\opt + \eps |E_{\adm}|$. Let $\calP' := \calP \cap \{ x : \cost(x) \leq \opt \}$ be the convex hull of good clusterings of cost at most $\opt$. $\calP'$ is non-empty since there exists a good clustering of cost at most $\opt$. Then we take an ellipsoid containing the $\calP'$. In every iteration, we take the center $x$ of the ellipsoid, and run the algorithm in Theorem~\ref{thm:round-or-cut} over this $x$. If the algorithm returns a clustering, then the cost of the clustering is at most $1.73\opt + \eps|E_{\adm}|$ and we are done. Otherwise, it returns a hyperplane separating $x$ and $\calP'$, which breaks the ellipsoid into two parts, one containing $x$, and the other containing $\calP'$. We define a new ellipsoid that contains the latter part, and repeat. Since $\calP'$ is nonempty, the algorithm will successfully output a desired clustering in polynomial number of iterations. Thus our goal in this section becomes to prove Theorem~\ref{thm:round-or-cut}.

A useful tool is the following lemma that shows we can try to extend $x$ to any domain. Suppose $D$ is any domain of variables, and $\calQ \subseteq [0, 1]^{{V \choose 2}} \times \R^D$ is a polytope such that the projection of $\calQ$ to coordinates in $V \choose 2$ contains $\calP$.  If $x$ is also in the projection, namely if we have $(x, y) \in \calQ$ for some $y \in [0, 1]^D$, then we can successfully extend $x$ to $(x, y)$ and use $y$ in our rounding algorithm.  Otherwise, we can output a separation plane for $x$ using the following lemma:
\begin{lemma}
    \label{lemma:find-cut}
    Suppose $x$ is not in the projection of $\calQ$; namely for every $y \in \R^D$, we have $(x, y) \notin \calQ$. Then in time polynomial in the description of $\calQ$, we can find a separation plane $(w, b)$ for $x$.
\end{lemma}
\begin{proof}
    Consider the LP $(A, A')\begin{pmatrix}x\\y\end{pmatrix} \leq c$ that defines the polytope $\calQ$, where $x$ corresponds to variables in $V \choose 2$ and $y$ corresponds to variables in $D$, and the number of columns of $A$ is the same as the number of rows of $x$. Now suppose for a fixed $x$, the LP is infeasible. That is, the LP $A'y \leq c - Ax$ (with variables $y$) is infeasible.  By LP duality, we can find a vector $u$ with non-negative entries such that $u^TA' = 0$ and $u^T(c - Ax) = -1$. On the other hand, for every $x'$ that is in the projection of $\calQ$,  we have $u^T(c - Ax') \geq 0$: There exists some $y$ with $(A, A')\begin{pmatrix}x'\\ y\end{pmatrix} \leq c$, which is equivalent to $A'y \leq c - Ax'$, which implies $0 = u^TA'y \leq u^T(c - Ax')$.  This gives us a separation plane between $x$ and $\calP$.  Moreover, the plane can be found in time polynomial in the size of the LP for $\calQ$.
\end{proof}

In the algorithm that proves Theorem~\ref{thm:round-or-cut}, we generate two clusterings using two different procedures ({\em set-based rounding} and {\em pivot-based rounding}) and output the better of the two clusterings. The properties of the two procedures are described in the following theorems.
In both theorems, the input is the same as that from Theorem~\ref{thm:round-or-cut}.
\begin{restatable}{theorem}{thmkt}{\rm (Set-based Rounding)}
    \label{thm:KT}
    There is an $n^{O(1/\eps^4)}$-time procedure $\calA_{\mathrm{set}}$ that either outputs a separation plane for $x$, or a clustering $\calC$ such that
    \begin{align*}
        \cost(\calC) \leq \sum_{vw \in E^+} \frac{2x_{vw}}{1+x_{vw}} \quad + \quad \sum_{vw \in E^-} \frac{1-x_{vw}}{1+x_{vw}}  \quad + \quad \eps \cdot |E_\adm|.
    \end{align*}    
\end{restatable}
\begin{restatable}{theorem}{cln}{\rm (Pivot-based Rounding)}
    \label{thm:CLN}
    There is an $n^{O(1/\eps^{2})}$-time procedure $\calA_{\mathrm{pivot}}$ that either outputs a separation plane for $x$ or a clustering $\calC$ such that
    \begin{align*}
        \cost(\calC) \leq \sum_{vw \in E^+} \min\{1.515 + x_{vw}, 2\} \cdot x_{vw} \quad + \quad 2\sum_{vw \in E^-}(1-x_{vw})  \quad + \quad \eps \cdot |E_\adm|. 
    \end{align*}
\end{restatable} 

It is simple to show that combining the two procedures proves Theorem~\ref{thm:round-or-cut}.

\begin{proof}[Proof of Theorem~\ref{thm:round-or-cut}]
Suppose that given $x \in [0, 1]^{\binom{V}{2}}$, both 
$\calA_{\mathrm{set}}$ and 
$\calA_{\mathrm{pivot}}$ return a solution $\calC_1$ and $\calC_2$ respectively; otherwise a separation plane is found. 

Consider the cost of $\calC_1$ times $0.42$ plus the cost of $\calC_2$ times $0.58$. Then the approximation ratio for any $+$edge is at most 
\begin{align}
\label{equ:final-ratio}
\max_{x \in [0,1]} 
\bigg( 0.42\frac{2}{1+x} + 0.58 \min(1.515+x, 2) \bigg) \leq 1.7257,
\end{align}
and the ratio for $-$edge is at most $0.42\cdot 1 + 0.58\cdot 2 = 1.58$. Therefore, the cost of the better one is at most $1.73\cost(x) + \eps |E_{\adm}| \leq 1.73\opt + \eps |E_{\adm}|$.

To see the ratio for $+$edges, notice that we only need to consider two $x$ values: $x = 0$ and $x=2-1.515=0.485$.  The function inside $\max(\cdot)$ in \eqref{equ:final-ratio} is decreasing for $x \in [0.485, 1]$ and convex for $x \in [0, 0.485]$.
\end{proof}

Theorem~\ref{thm:KT} and~\ref{thm:CLN} will be proved in Section~\ref{sec:KT} and~\ref{sec:CLN} respectively.

    \section{Preclustering}
\label{sec:preprocess}
In this section, we show how to find a preclustered instance and prove Theorem~\ref{thm:preprocessed-wrapper}. 
Given a Correlation Clustering instance $I = (V, E^+ \cup E^-)$ on a set of vertices $V$, we let $G = (V, E^+)$ be the graph of $+$edges, and $d_v$ and $N_v$ be respectively the degree and neighbor set of $v$ in $G$.  For reasons that we will note later on, 
 we assume that each vertex has a self-loop \pedge. (Note that this does not affect the cost of clustering, as
a self-loop is never cut by a clustering.)  This assumption implies that $v \in N_v$, and we have $d_v =|N_v|$.  Moreover, the number of $+$neighbors of $v$ equals $d_v$.  Notice that without loss of generality, we can assume that each vertex $v$ in the input instance has at least one proper $+$neighbor (i.e., $d_v \geq 2$); otherwise, we would clearly put that vertex in its own (singleton) cluster and solve the remaining instance.
Let $\epsq < 10^{-8}$ be a fixed constant, and let $\epsfour = \sqrt{\epsq}$.

    The goal of this section is to prove the following theorem. 
    \begin{theorem}
    \label{thm:preprocessed} 
        Given a Correlation Clustering instance $G$, we can in polynomial time construct a preclustered instance $(\calK, E_\adm)$ 
        with the following two properties:
        \begin{itemize}
            \item $\cost(\calC^*_{\prepinstance})$ is at most $(1+\epsfour)$ times the cost of the optimum clustering for $G$,
            \item $\cost(\calC^*_{\prepinstance})$ is at least $(\epsq^6/2) \cdot |E_{\adm}| = \epsa |E_{\adm}|$,
        \end{itemize}
        where $\calC^*_{\prepinstance}$ denotes a good (but unknown) clustering for $(\calK, E_\adm)$ of minimum cost.
    \end{theorem}
Note that Theorem~\ref{thm:preprocessed-wrapper} immediately follows from Theorem~\ref{thm:preprocessed} by letting $\eps = \sqrt{\eps_q} \leftarrow \delta$. 
  For the rest of the section, we prove Theorem~\ref{thm:preprocessed} by describing a procedure that takes a Correlation 
Clustering instance and outputs a preclustered instance.

\paragraph{Atomic Subpartitioning}
Our algorithm is similar to the algorithm of \cite{CLMNP21} and relies on the 
notion of weak agreement. Informally, we say that $u$ and $v$ are in agreement when their 
neighborhoods are almost identical, up to a tiny fraction of the neighbors.
We expect $u$ and $v$ to be treated similarly in a target optimal solution: either $u$ 
and $v$ are in the same cluster, or both form singleton clusters.

Our algorithms are parameterized by two constants $\beta$, $\lambda$ that will be determined later.
\begin{definition}[Weak Agreement~\cite{CLMNP21} ]\label{definition:agreement}
    Two vertices $u$ and $v$ are in \emph{$i$-weak agreement} if $|N_u \triangle N_v| < i \beta \cdot \max\{|N_u|, |N_v|\}$, where $N_u \triangle N_v$ denotes the symmetric difference between $N_u$ and $N_v$. If $u$ and $v$ are in $1$-weak agreement, we also say that $u$ and $v$ are in \emph{agreement}.
\end{definition}

\begin{remark} As noted earlier, we assume that each vertex has a self-loop \pedge, because we follow the assumptions of 
\cite{CLMNP21}.
\end{remark} 

\cite{CLMNP21} then provide Algorithm \ref{alg:admissible}, which we
apply with $\beta := \epsq$ and $\lambda  := \epsq$.
\begin{algorithm}[h!]
\begin{algorithmic}[1]
\caption{Atomic-Preclustering($G$) -- Algorithm 1 in~\cite{CLMNP21}\label{alg:admissible}}
  \State Discard all \pedges whose endpoints are not in agreement. \;
  \State Call a vertex \emph{light} if it has lost more than a $\lambda$-fraction of its $+$neighbors in the previous step. Otherwise call it \emph{heavy}.
    \State Discard all \pedges between two light vertices. 
  \State Let $\tG$ be the \emph{sparsified graph} on the remaining \pedges. 
  Let $C_1,\ldots, C_k$ be the connected components of $\tG$; each $C_i$ with $|C_i| \ge 2$ is defined as 
  an \emph{atom}.
  Any pair $u,v$ such that $u$ and $v$ are in the same atom is an atomic pair (or edge).
\end{algorithmic}
\end{algorithm}

We will make use of the following structural lemmas. 
\begin{lemma}[Lemma 3.4 in full version (arXiv) of \cite{CLMNP21}]
\label{lem:quasiclique-degree}
Let $\beta =\lambda =\epsq < 1/100$ in Algorithm~\ref{alg:admissible}.
    Let $C$ be a connected component of $\tG$ of size at least 2, then for each vertex $u$ in $C$,
    we have that the number of $+$neighbors of $u$ in $C$ is at least $(1-9\epsq)|C|$.
\end{lemma}
\begin{lemma}[Lemma 3.3 in full version (arXiv) of \cite{CLMNP21}]
\label{lem:quasiclique-neighbors}
Let $\beta =\lambda =\epsq < 1/100$ in Algorithm~\ref{alg:admissible}.
Then, for any $u,v$ in the same atom, we have that if $u$ or $v$ is heavy, 
then $u$ and $v$ are in $4$-weak agreement. Moreover, for any atom $K$,
the hop-distance in $K$ between any two vertices of $K$ is at most 4.
\end{lemma}
\begin{lemma}[Fact 3.2 of~\cite{CLMNP21}, second bullet]
\label{lem:transitivityagreement}
Let $k \in \{2, 3, 4, 5\}$ and $v_1 ,\ldots , v_k \in V$ be a sequence of vertices such that $v_i$ is in agreement with $v_{i+1}$ for $i \in \{1, \ldots , k - 1$\}. 
Then $v_1$ and $v_k$ are in $k$-weak agreement.
\end{lemma}
One of our key lemmas is based on Lemma 3.5 of the full arXiv version of~\cite{CLMNP21}.
\begin{lemma}[Atom Structure in Preclustering]
\label{lem:quasiclique}
    There exists an optimum solution $\calC_1$ for $I$ such that for any atom $K$, there exists a cluster 
    $C \in \calC_1$ containing $K$. 
\end{lemma} 
\begin{proof}
Consider an atom $K$. By the description of Algorithm~\ref{alg:admissible}, 
there is a heavy vertex $v$ in $K$, because $K$ contains at least one edge and edges with two light endpoints have been removed.  Since $v$ is heavy, 
a $(1-\epsq)$-fraction of its $+$neighbors are also in $K$.  Moreover by Lemma~\ref{lem:quasiclique-neighbors}, every other vertex $u \in K$ is in 4-agreement with $v$.  The vertex $v$ has at most $\epsq|K|$ $+$neighbors outgoing from $K$.  By 4-agreement, for every other vertex $u \in K$, $u$ has at most $36 \epsq |K|$ $+$neighbors outgoing from $K$ (by the choice of $\epsq$).  
Moreover, by Lemma~\ref{lem:quasiclique-degree}, $u$ has at most $9\epsq |K|$ $-$neighbors in $K$.

Now, fix an optimum solution $\calC$ and assume toward contradiction
that there exists $C_1,\ldots,C_{\ell} \in \calC$, where $\ell > 1$,
such that $C_i \cap K \neq \emptyset$ for all $i$.  Recall that each
vertex $u \in K$ has at most $36\epsq|K|$ + neighbors outside of $K$
and so for each $C_i$, $|C_i\setminus{K}| \le 73\epsq |K|$ since
otherwise the number of \medges between $C_i\setminus{K}$ and $C_i
\cap K$ outnumbers the number of \pedges and the clustering where
$C_i$ is replaced with $C_i\setminus{K}$ and $C_i \cap K$ has cheaper
cost.

Similar to the proof of Lemma 3.5 in~\cite{CLMNP21}, we consider two
different cases, either there exists a cluster $C^*$ of size larger
than $(1-200\epsq)|K|$ or all the clusters have size at most
$(1-200\epsq)|K|$.

In the first case, observe that by the above discussion, $C^*$
contains at least $(1-273\epsq)|K|$ vertices of $K$ and so for each
$C_i \neq C^*$, we have that $|C_i \cap K| \le 273\epsq|K|$.  Since
each vertex $v \in C_i\cap K$ has at least $(1-9\epsq)|K|$ + neighbors
in $K$, we have that it has at least $(1-282\epsq)|K|$ + neighbors in
$C^*$.  Moreover, since $|C^*| \le (1+73\epsq)|K|$ by the above
discussion, the cost of moving $v$ from $C_i$ to $C^*$ is at most
$619\epsq|K|$ while the saving is at least $(1-282\epsq)|K|$ and since
$\epsq < 1/1000$, the change induces a positive saving and so a
contradiction to the fact that the solution is optimal.

We thus turn to the second case where all the clusters
$C_1,\ldots,C_{\ell}$ have size at most $(1-200\epsq)|K|$.  By Lemma
\ref{lem:quasiclique-degree}, each vertex $v \in C_i\cap K$ has at
least $(1-9\epsq)|K| - |C_i| \ge (1-9\epsq)|K| - (1-200\epsq)|K| =
191\epsq|K|$ + neighbors in $K\setminus{C_i}$. Hence $\calC$ separates
at least $191|K|^2\epsq/2$ \pedges.  On the other hand, consider
modifying $\calC$ by removing all the vertices in $K$ from their
current clusters and creating a cluster consisting only of $K$. The
new clustering will be saving at least $191|K|^2\epsq/2$ (by the above
computation) while paying an additional $10|K|^2\epsq$ cost for the
internal - edges and an additional $36\epsq|K|^2$ cost for the
outgoing + edges.  Since $46\epsq|K|^2 < 191|K|^2\epsq/2$, we have
that the resulting clustering is of cheaper cost, a contradiction that
concludes the proof.
\end{proof}

\paragraph{Admissibility}

A pair $(u,v)$ is \emph{degree-similar} if $\epsq d_v \le d_u \le d_v
/ \epsq$.  Let $\calK$ be the set of atoms output by
Algorithm~\ref{alg:admissible}.  A pair $(u,v)$ is \emph{admissible}
if (1) either $u$ or $v$ belongs to $V\setminus{V_\calK}$, and (2) it
is degree-similar, and (3) the number of common neighbors that are
degree-similar to both $u$ and $v$ is at least $\epsq \cdot
\text{min}\{d_u, d_v\}$.  We let $E_\adm$ be the set of all admissible
pairs in $V \choose 2$.  This finishes the description of our
preclustered instance $(\calK, E_\adm)$.  To this end, we let $G' =
(V, E')$ be the graph containing the set of atomic and admissible
edges, and for every $v \in V$, let $N'_v$ and $d'_v$ be the neighbor
set and degree of $v$ in $G'$.

\begin{lemma} 
\label{lem:num-admissible}
    The following two properties hold for $(\calK, E_\adm)$:
    \begin{enumerate}[label=(\roman*)]
         \item for every $v \in V$, we have $d'_v \le 2\epsq^{-3} \cdot d_v$, and 
        \item for every $uv \in E'$, we have $d_u \le 2 \epsq^{-1} \cdot d_v$.
    \end{enumerate} 
\end{lemma}

\begin{proof} 
We consider an arbitrary vertex $u$.  Let us first analyse the number
of edges $u,v$ that are admissible edges.  By the definition of an
admissible pair, we have that $d_u \le \epsq^{-1} d_v$ for an
admissible edge $uv \in E'$.  Next we show that for any $u$, $d'_u \le
d_u/\epsq^3$.  Let $D_u$ be the set of neighbors of $u$ in $G$ that
are degree similar to $u$.  Each vertex $v$ such that $(u,v)$ is an
admissible pair is connected to at least $\epsq \cdot \min\{d_u,d_v\}
\geq \epsq^2 d_u$ vertices in $D_u$.  Moreover, each vertex $w \in
D_u$ has degree $d_w \leq d_u/\epsq$ by definition of
degree-similar. Hence, by using a counting argument, we conclude that
the total number of vertices $v$ such that $uv$ is in $E_{\adm}$ is at
most $|D_u|(d_u/\epsq) / (\epsq^2 d_u) \le d_u/\epsq^3$.  This yields
the desired bounds for any vertex $u$ that does not belong to an atom.

    To finish the analysis, we need to analyze the number of
    atomic edges attached to any vertex $u$ in an atom
    and show that a pair of vertices $u,v$ in the same atom satisfies
    $d_u \le 2\epsq^{-1} \cdot d_v$. 
  By Lemma~\ref{lem:quasiclique-neighbors} and
    Lemma~\ref{lem:transitivityagreement}, we have that $u$ is in
    8-weak agreement with any element of $K$ so we have that $d_u \le
    2 d_v \cdot \epsq^{-1}$ by our choice of $\epsq$, as desired. 
    Moreover,
    consider a heavy vertex $v'$ in $K$ (there must exist one by
    definition).
    The fact $v$ remained heavy implies that $K$ contains at least a $(1-\epsq)$ fraction of the $+$neighbors of $v'$, and Lemma~\ref{lem:quasiclique-degree} implies that the degree of $v'$ is at least
    $(1-O(\epsq))|K|$.
    Since any other
    vertex is in 4-weak agreement with $v'$, we have that the degree
    of any vertex $u \in K$ is $(1\pm O(\epsq))|K|$ and so $d'_u \le
    2\epsq^{-3} \cdot d_u$.
\end{proof}

We now want to prove Theorem~\ref{thm:preprocessed} for an  preclustered instance $(\calK, E_{\adm})$.
We start with the following lemma.
\begin{lemma}
\label{lem:neighbor-constrained-sol}
    Let $\calK$ be the set of atoms output by Algorithm~\ref{alg:admissible}.
    There exists a good clustering $\calC_2$ of $I$ with cost at most $(1+\eps)$ times the optimum solution of
    $I$ such that 
    \begin{enumerate}
    \item For any vertex $v$ that belongs to a non-singleton cluster $C \in \calC_2$, we have 
        that $v$ has at least $(1+\eps)|C|/2$ $+$neighbors in $C$.
\item For any vertex $v$ that belongs to a non-singleton cluster $C \in \calC_2$, we have       that $|C| \ge \eps d_v$.
    \end{enumerate}

\end{lemma}
\begin{proof}
    Start from $\calC_1$ as per Lemma \ref{lem:quasiclique} and modify it as follows.
   For any cluster $C \in \calC_1$, let $\calK(C)$ be the set of atoms $K$ such 
   that $K\cap C \neq \emptyset$ (note that by Lemma~\ref{lem:quasiclique}, if
   $K\cap C \neq \emptyset$, then $K \subseteq C$).
   For any cluster $C \in \calC_1$, if   
   the number of \medges within $C$ is at least $(1-10\eps)\binom{|C|}{2}/2$, 
   then make each $K \in \calK(C)$ a cluster by itself and each non-atom vertex $v \in C$
   a singleton cluster. This 
   increases the cost by a factor of at most $(1+10\eps)/(1-10\eps)$.
   Secondly, if the number of \pedges
    outgoing from the cluster is at least $|C|^2/\eps$, then make all the non-atom
    vertices of 
    $C$ singleton clusters and each $K \in \calK(C)$ a cluster by itself, 
    and charge the cost of $|C|^2$ for the operation to the 
    outgoing \pedges); each edge across clusters of $\calC_1$ each gets 
    a charge of $O(\eps)$ ($\eps$ for each endpoint).
    Let $\calC_1'$ be the resulting solution which has cost at most
    $(1+\eps)$ times the cost of $\calC_1$ by the above argument.

    We next apply the following iterative procedure to $\calC_1'$: As long as there exists a vertex $v$ in a 
    non-singleton cluster $C$ such that either Event (1): $v$ has less than $(1+\eps)|C|/2$ neighbors in $C$, or Event (2): $|C| < \eps d_v$, we proceed as follows: 
    \begin{itemize}
    \item\textbf{Individual case} If $v$ is not in an atom, we make $v$ a singleton.
    \item\textbf{Atom case} Otherwise $v$ belongs to an atom, we remove the 
    whole atom $K$ containing 
    $v$ from $C$ and make it a cluster on its own. 
    \end{itemize}

    We first argue that the procedure outputs a clustering that satisfies the 
    conditions of the lemma. 
    The only way the two items can be violated is by an atom 
    (since otherwise, the problematic vertex would be made a singleton).
    We thus argue
    that atoms satisfy the two properties.
    By Lemma~\ref{lem:quasiclique-degree} and $\eps < 1-18\epsq$, each atom satisfies the first bullet.
    We then argue that each atom $K$ is such that for any vertex $v \in K$, 
    $\eps d_v < |K|$. This follows from the fact that each atom $K$ contains a
    heavy vertex $u$ that is thus in agreement with a $(1-\epsq)$ fraction of its 
    neighbors that hence join the atom; it follows that $|K| > (1-\epsq) d_u > \eps d_u$. 
    For the the other vertices, $v$ is in 4-weak agreement with $u$, and so the property holds.
    It follows that when the procedure stops, the solution obtained satisfies
    the constraints of the theorem. 

    We then need to account for the cost increase. We first make the following 
    observation: Any non-singleton cluster $C \in \calC_1'$ (and so of $\calC_1$) 
    cannot lose more than $(1-\eps)|C|$ vertices because of Event (1) and $|C|/2$ vertices because of Event (2).
    Indeed, 
    if it loses more than $(1-\eps)|C|$ vertices because of Event (1), then that implies
    that the total number of \medges internal to $C$ was already larger than 
    $(1-10\eps)|C|(|C|-1)/4$ and should not have been in $\calC'_1$. Moreover, if it 
    loses more than $|C|/2$ vertices because of Event (2), that means that the number
    of \pedges outgoing $C$ was at least $|C|^2/(2\eps)$ and should not have been in $\calC'_1$.

We can now analyse the cost increase due to Event (2). In the
individual case, namely when a single vertex $v$ is placed in a
singleton cluster, we have that $v$ has at least $|C|(1-\eps)/(2\eps)$
$+$neighbors outside its cluster in $\calC'_1$ and so the cost can be
charged (with an individual edge charge of $2\eps$) to these \pedges.
In the atom case, namely when an entire atom $K$ is placed in a single
cluster, we have that there exists a vertex $v$ in $K$ with at least
$(|C|-1)/(2\eps)$ $+$neighbors outside its cluster in
$\calC'_1$. Moreover, each vertex $u \in K$ is in at most
$5$-weak-agreement (by Lemmas~\ref{lem:quasiclique-neighbors}
and~\ref{lem:transitivityagreement} and so has at least
$(1-2\eps)(|C|-1)/(2\eps)$ $+$neighbors outside its cluster in
$\calC'_1$ and so the cost of moving the atom outside of $C$ can be
charged to the $+$neighbors of the atom that are not in the cluster of
the atom in $\calC_1'$, with a charge of $2\eps/(1-2\eps)$ per
individual edge. Therefore, the cost of Event (2) can be charged to
the \pedges paid in $\calC'_1$, each edge receiving a charge of
$O(\eps)$.
    
    It thus remains to bound the cost incurred by Event (1).
    Consider the atom case, namely the case where an entire atom $K$ is 
    moved out of a cluster and creates a cluster on its own. In this case, we have
    that there is a vertex $u \in K$ which is adjacent to less than a $(1+\eps) |C|/2$
    fraction of its current cluster $C$ and assume that $\eps d_u \le |C|$ since
    otherwise this is an Event (2) case.
    Note that in this case, since  $\eps d_u \le |C|$ and
    by Lemma~\ref{lem:quasiclique-neighbors} that for any vertex $u \in K$, 
    $|N_u \triangle N_v| < 4\epsq \cdot\max\{|N_u|, |N_v|\}$ and so each vertex $u \in K$
    is adjacent to less than $(1+5\eps)|C|/2$ vertices in $C$.
    This implies that whenever a vertex $v$ is moved out of its cluster in the procedure,
    it is adjacent to at most $(1+5\eps)|C|/2$ vertices in $C$ and so has at least 
    $(1-5\eps)|C|/2$  $-$neighbors in $C$. 
    We thus charge the cost of moving $v$ outside of $C$, which is at most $(1+5\eps)|C|/2$, to the cost $v$ was paying in $\calC'_1$ by placing a charge 
    of $O(\eps)$ on each $-$neighbor of $v$ in $C$. This account for the change in cost 
    since each $-$neighbor of $v$ in $C$ was contributing one to the objective and is now
    contributing 0. Note that in this way, each \medge that contributes to the objective
    in $\calC'_1$ is charged at most once: when one of its endpoints is moved out of 
    the cluster and cannot be charged later in the procedure since the endpoints remain
    in different clusters.

\end{proof}

We can now turn to the proof of Theorem~\ref{thm:preprocessed}.

\begin{proof}[Proof of Theorem~\ref{thm:preprocessed}]
We consider the solution $\calC_2$ from Lemma~\ref{lem:neighbor-constrained-sol} and the pair $\prepinstance$.
It is easy to see that $\prepinstance$ can be computed in polynomial time.

Lemma~\ref{lem:neighbor-constrained-sol} already shows that the cost of $\calC_2$ is within a $(1+\eps)$ factor of the optimum cost and so it only remains to show that $\calC_2$ is a good clustering for $\prepinstance$  and that its cost is at least $\epsa \cdot |E_{adm}|$.

Let us first argue that $\calC_2$ is a good clustering for $\prepinstance$, namely that 
\begin{itemize}
    \item For every atom $K\in \calK$, we have $K \subseteq C$ for some $C \in 
    \calC_2$.
    \item For every non-admissible edge $uv$, $u$ and $v$ are not in the same cluster in 
    $\calC_2$. 
\end{itemize}

Note that the first bullet is satisfied by
$\calC_2$ by Lemma~\ref{lem:neighbor-constrained-sol}. It thus remains to show that the second bullet is satisfied.
By Lemma~\ref{lem:neighbor-constrained-sol}, pairs of vertices in the same cluster are degree-similar and moreover, for any vertex $v$ that belongs to a non-singleton cluster $C \in \calC_2$, 
we have that $v$ has at least $(1+\eps)|C|/2$ $+$neighbors in $C$ that are degree-similar.  This implies that any pair of vertices $u,v$ that are in the same cluster $C \in \calC_2$ has at least $\eps|C|$
common neighbors that are degree similar,
and by Lemma~\ref{lem:neighbor-constrained-sol},  $\eps|C| \geq \epsq \cdot \min\{d_u, d_v\}$.
Thus, either $u,v$ is a pair of vertices in the same atom (hence $(u,v) \in E'$) or it is admissible (and so $(u,v) \in E'$ too) as desired.

We now turn to providing a lower bound on the cost.
We say that a vertex $v$ is a \emph{high-contributor} if its contribution to the cost is at least $\epsq^2 d_v$.  Thus, any vertex that is not a high-contributor must be adjacent to a $(1-\epsq^2)$-fraction of its cluster and have at most $\epsq^2 d_v$ $+$neighbors outside of its cluster. It follows that $+$neighbors that are not high-contributors and that are in the same cluster are in agreement. 
Next, define a cluster $C$ to
be a high-contributor cluster if the total number of outgoing \pedges (i.e., \pedges with 
exactly one endpoint in $C$) is at least 
$\epsq^4 |C|^2$ or the total number of internal \medges (\medges with both endpoints in $C$)
is at least $\epsq^4 |C|^2$. 

Next, consider a non-high-contributor cluster $C$; it must be that the number of non-high-contributor vertices in $C$ is at least $(1-\epsq^2)|C|$.
Thus, since each non-high-contributor vertex in $C$ is adjacent to at least $(1-\epsq^2)|C|$ vertices in $C$ and
in agreement with the other non-high-contributor vertices, we have that each non-high-contributor vertex is heavy, and so all the non-high-contributor vertices of $C$ are
in the same atom. 

We can thus rewrite a lower bound on the cost of the solution $\calC_3$ as 
\[\sum_{\text{high-contributor cluster } C} \epsq^2|C|^2 + \sum_{\text{non high-contributor cluster } C} \sum_{\text{high-contributor vertex } v \in C} \epsq^2 d_v\]

\begin{eqnarray*}
\sum_{\text{high-contributor cluster } C} \epsq^2|C|^2 & \ge & 
\sum_{\text{high-contributor cluster } C} \epsq^2 \sum_{v \in C}  |C| \\
& \geq &
\sum_{\text{high-contributor cluster} C} \epsq^2 \sum_{v \in C}  \eps d_v \\ &\geq&
\sum_{\text{high-contributor cluster } C} \sum_{v \in C} \epsq^3 d_v,
\end{eqnarray*}
where the second inequality follows from the second bullet of Lemma
\ref{lem:neighbor-constrained-sol}.

So we have that the total contributions of vertices is at least:
\[\sum_{\text{non high-contributor cluster } C} \sum_{v \in C\setminus{K}} \epsq^3 d_v + 
\sum_{\text{high-contributor cluster } C} \sum_{v \in C} \epsq^3 d_v + \sum_{\text{Singletons}} d_v/2 \geq \sum_{v \in V\setminus{V_\calK}} \epsq^3 d_v\]

It follows that the cost of solution $\calC_3$ is at least 
$\epsq^{3} \sum_{v \in V \setminus V_{\calK}} d_v$, as claimed (where $\epsq = \eps^2$).

By Lemma \ref{lem:num-admissible}, we have 
\[\sum_{v \in V\setminus{V_\calK}} \epsq^3 d_v \geq 
\sum_{v \in V\setminus{V_\calK}} \epsq^3 (d'_v \epsq^3)/2  \geq 
\sum_{v \in V\setminus{V_\calK}} \epsq^6 d'_v/2 \geq \epsq^6/2 \cdot |E_{\adm}|. \qedhere
\]
\end{proof}

    \section{Set-Based Rounding Procedure} \label{sec:KT}
In this section, we prove Theorem~\ref{thm:KT}, by giving the set-based rounding algorithm $\calA_{\mathrm{set}}$. We repeat the theorem below:
\thmkt*

Recall that we are given a Correlation Clustering instance $(V, E^+ \cup E^-)$, a preclustered instance $(\calK, E_\adm)$, and a metric $x \in [0, 1]^{V \choose 2}$ satisfying the properties in Theorem~\ref{thm:round-or-cut}. The output is either a separation plane between $x$ and $\calP$, or a clustering $\calC$. Recall that $\calP$ is the convex hull of the metrics of all good-clusterings for $(\calK, E_\adm)$.  

Throughout the section, we use $N_\adm(u)$ to denote the set of admissible edges incident to $u$, $d_\adm(v)$ to denote its size, and $E_\adm(A, B)$ for two disjoint subsets $A$ and $B$ to denote the set of admissible edges between $A$ and $B$. It is good to keep in mind the algorithm for the $O(1)$-sized cluster case from the introduction, as it will serve as a baseline for our more general algorithm.

\subsection{Linear Program Relaxation}
During the algorithm, we maintain the set $V'$ of vertices that are not clustered yet, and let $n' = |V'|$; initially we have $V' = V$. 
For $(\calK, E_\adm)$ and for each $u \in V$, we shall let $K_u$ denote the atom containing $u$, if there is one; otherwise let $K_u = \{u\}$.
Let $\calC^*$ be any good clustering for $(\calK, E_\adm)$; in other words, the metric for $\calC^*$ is a vertex-point of $\calP$. Let $\calC'^*$ be the clustering $\calC^*$, restricted to $V'$ and with empty clusters removed.  We shall define a clustering $\tilde \calC$ which will serve as the target clustering for our LP.  The cluster $\tilde \calC$ is constructed as follows:
\begin{algorithmic}[1]
    \State let $\tilde \calC \gets \calC'^*$
    \While {there exists some $K_u$ in a cluster $C \in \tilde \calC$ with $|K_u| < |C| < \eps d_\adm(u) + |K_u|$}
        \State $\tilde \calC \gets \tilde \calC \setminus \{C\} \cup \{K_u, C \setminus K_u\}$
    \EndWhile
\end{algorithmic}
Note this procedure is only for analysis purpose and is not a part of our algorithm, as we do not know $\calC^*$ and $\calC'^*$.

\begin{claim}
    \label{claim:tilde-C}
    The following statements are true for the clustering $\tilde \calC$:
    \begin{enumerate}[label = (\ref{claim:tilde-C}\alph*),leftmargin=*]
        \item For every $u \in V'$, $K_u$ is either a cluster, or in a cluster of size more than $|K_u| + \eps\cdot d_\adm(u)$. 
        \item The number of pairs in $V' \choose 2$ cut in $\tilde \calC$ is at most that in $\calC'^*$ plus $\eps \cdot \sum_{u \in V'} d_\adm(u)$.
    \end{enumerate}
\end{claim}
\begin{proof}
    The first statement is straightforward. Whenever we break $C$ into $K_u$ and $C \setminus K_u$ in the procedure, the cost increase is at most $|K_u| \cdot |C \setminus K_u| \leq |K_u| \cdot \eps d_\adm(u) \leq \eps \sum_{v \in K_u} d_\adm(v)$. We separate each $K_u$ at most once. Therefore, the second statement holds. 
\end{proof}

 We then describe the LP relaxation, which depends on $V'$. Suppose we have a good clustering $\calC^*$, where $\calC'^*$ is the clustering $\calC^*$ restricted to $V'$, and let $\tilde \calC$ be obtained as above.  Let $r=\Theta(\frac1{\eps^4})$, with a large enough hidden constant inside $\Theta(\cdot)$. In the LP, we have a variable $y^s_S$, for every $s \in [n']$, and $S \subseteq V'$ of size at most $r$, that denotes the number of clusters of size $s$ in $\tilde \calC$ containing $S$ as a subset. When $S \neq \emptyset$, there is at most one such cluster and thus $y^s_S \in \{0, 1\}$ in an integral solution. For every $S$, let $y_S := \sum_s y^s_S$ denote the number of clusters (of any size) in $\tilde \calC$ containing $S$ as a subset. If $S \neq \emptyset$, then $y_S$ indicate if $S$ is a subset of a cluster in $\tilde \calC$ or not. For every $uv \in {V' \choose 2}$, we have a variable $\tilde x_{uv}$ indicating if $u$ and $v$ are separated or not in $\tilde \calC$.

For convenience, we shall use the following type of shorthand:  $y^s_{u}$ for $y^s_{\{u\}}$, $y^s_{uv}$ for $y^s_{\{u, v\}}$, and $y^s_{Su}$ for $y^s_{S \cup \{u\}}$. The LP is defined as in LP(\ref{LPC:KT-define-s}-\ref{LPC:KT-non-negative}). In the description of the LP, we always have $s \in [n'], u \in V'$ and $uv \in {V' \choose 2}$. For convenience, we omit the restrictions.    By default, any variable of the form $y_S$ or $y^s_S$ has $|S| \leq r$; if not, we do not have the variable and the constraint involving it. 

\noindent\begin{minipage}[t]{0.34\textwidth}
    \begin{align}
    	\sum_{s = 1}^{n'} y^s_S &= y_S   & &\forall S \label{LPC:KT-define-s}\\
    	y_u &= 1 & &\forall u \label{LPC:KT-a-contained}\\
    	y_{uv} &= 1 - \tilde x_{uv} & &\forall uv \label{LPC:KT-define-tx}\\
    	\tilde x_{uv} &\geq x_{uv} & &\forall uv \label{LPC:KT-tx-to-x}\\
    	\frac1s\sum_{u} y^s_{Su} &=  y^s_S &\quad &\forall s, S \label{LPC:KT-size-s}
    \end{align}
\end{minipage}\hfill
\begin{minipage}[t]{0.62\textwidth}
    \begin{align}
        \tilde x_{uv} &= 0 & &\forall u, v \text{ in a same } K \in \calK\\
    	\sum_{uv} (\tilde x_{uv} - x_{uv}) &\leq \eps\sum_{u} d_\adm(u) \label{LPC:KT-cost}\\
    	y^s_S &= 0 & &\text{if implied by (\ref{claim:tilde-C}a)} \label{LPC:KT-forbid}\\
        \sum_{T' \subseteq T}(-1)^{|T'|}y^s_{S\cup T'} &\in [0, y^s_S] & &\forall s, S\cap T =\emptyset
        \label{LPC:KT-correlation}\\
    	\text{all variables }&\text{are non-negative} \label{LPC:KT-non-negative}
    \end{align}
\end{minipage}\medskip

$\tilde x$ and $y$ variables are the LP variables, and $x$ variables are given as the input by Theorem~\ref{thm:KT}.  \eqref{LPC:KT-define-s} gives the definition of $y_S$, \eqref{LPC:KT-a-contained} requires $u$ to be contained in some cluster, and \eqref{LPC:KT-define-tx} gives the definition of $\tilde x_{uv}$. \eqref{LPC:KT-tx-to-x} holds as $\tilde \calC$ is a refinement of $\calC'^*$.  
\eqref{LPC:KT-size-s} says if $y^s_S = 1$, then there are exactly $s$ elements $u \in V$ with $y^s_{Su} = 1$. (An exception is when $S = \emptyset$; but the equality also holds). \eqref{LPC:KT-cost} follows from (\ref{claim:tilde-C}b) and \eqref{LPC:KT-forbid} is by (\ref{claim:tilde-C}a).  The left side of \eqref{LPC:KT-correlation} is the number of clusters of size $s$ in $\tilde \calC$ containing $S$ but does not contain any vertex in $T$. So the inequality holds. This corresponds to a Sherali-Adams relaxation needed for the correlated rounding~\cite{RT12}, see Lemma~\ref{lem:RT-KT}. \eqref{LPC:KT-non-negative} is the non-negativity constraint.

If $x \in \calP$, then the above LP is feasible, as this holds for
every vertex point $x$ of $\calP$.  On the other hand, if the LP is
infeasible, then we can use Lemma~\ref{lemma:find-cut} to return a
separation plane between $x$ and $\calP$.

\subsection{Rounding Algorithm}
With the LP defined, we can then describe the rounding algorithm for
Theorem~\ref{thm:KT}.  The pseudo-code is given in
Algorithm~\ref{alg:KT-rounding}, and it calls
Algorithm~\ref{alg:KT-construct-C}.

\begin{algorithm}
	\caption{Set-Based Rounding Procedure}
	\label{alg:KT-rounding}
	\begin{algorithmic}[1]
		\Require{A Correlation Clustering instance $G$, a preclustered instance $(\calK, E_\adm)$ for $G$, and a metric $x \in [0, 1]^{V \choose 2}$ satisfying properties in Theorem~\ref{thm:round-or-cut}, and $\eps > 0$}
		\Ensure{Either a separation plane for $x$, or an integral clustering $\calC$}
		\State $V' \gets V, \calC \gets \emptyset$
		\While{$V' \neq \emptyset$}
		\State try to solve LP(\ref{LPC:KT-define-s}-\ref{LPC:KT-non-negative}) for the $V'$ to obtain a vector $(\tilde x, y)$
		\If {the LP is infeasible}
		\State \Return a separation plane for $x$ using Lemma~\ref{lemma:find-cut}
		\Else
		\State $C \gets$set-based-cstr-clst$(V', y)$, $\calC \gets \calC \cup \{C\},V' \gets V' \setminus C$
		\EndIf
		\EndWhile
		\State \Return $\calC$
	\end{algorithmic}
\end{algorithm}

\begin{algorithm}
	\caption{set-based-cstr-clst$(V', y)$}
	\label{alg:KT-construct-C}
	\begin{algorithmic}[1]
		\State randomly choose a cardinality $s$, so that $s$ is chosen with probability $\frac{y^s_\emptyset}{y_\emptyset}$
		\State randomly choose a vertex $u \in V'$, such that $u$ is chosen with probability $\frac{y^s_u}{s y^s_\emptyset}$
		\State define a vector $y'$ such that $y'_S = \frac{y^s_{Su}}{y^s_u}$ for every $S \subseteq V$ of size at most $r-1$
		\State\label{step:KT-RT}apply the Raghavendra-Tan correlated rounding technique over the fractional set $y'$ to construct the cluster $C \subseteq V'$ that does not break any atom, and \Return $C$
	\end{algorithmic}
\end{algorithm}
Step \ref{step:KT-RT} of Algorithm~\ref{alg:KT-construct-C} uses the following {\em correlated rounding} from Raghavedra and Tan~\cite{RT12}.
Note that vertex $v$ with $y'_v \in \{ 0, 1 \}$ is trivially decided and we incur errors only for $v$ with $y'_v \in (0, 1)$. 

\begin{lemma}[\cite{RT12}]
\label{lem:RT-KT}
  In Step~\ref{step:KT-RT} of Algorithm~\ref{alg:KT-construct-C}, let $V'_u := \{ v : y'_v \in (0, 1) \}$.   
  One can sample $C \subseteq V'$ in time $n^{O(r)}$ such that
  \begin{itemize}
  \item For each $v \in V'$, $\Pr[v \in C] = y'_{uv}$. 
  \item $\E_{u, v \in V'_u}[|\Pr[v, w \in C] - y'_{uvw}|] \leq \epsrt$, where $\epsrt = O(1/\sqrt{r})$.
  \end{itemize}
\end{lemma}
Recall that $r = \Theta\left(\frac1{\eps^4}\right)$. So by setting the constant appropriately, we can have $\epsrt = \eps^2$.

\subsection{Analysis of Error-Free Version of Algorithm~\ref{alg:KT-construct-C}}
In this section, we analyze Algorithm~\ref{alg:KT-construct-C} by ignoring the errors incurred by the RT procedure. We focus on one execution of the algorithm with input $V'$ and $y$. To do this formally, we define $\err^s_{vw|u}$ to be the error generated by the procedure when we choose $s$ as the size and $u$ as the pivot:
\begin{align*}
    \err^s_{vw|u} := \left|\Pr\big[v, w \in C|s, u\big] - \frac{y^s_{uvw}}{y^s_u}\right|, \forall vw \in {V' \choose 2},
\end{align*}
and 
\begin{align*}
	\err^s_{vw}:=\frac{1}{sy^s_\emptyset}\sum_{u \in V'}y^s_u\cdot\err^s_{vw|u}
	\qquad \text{and}\qquad \err_{vw} := \sum_{s} \frac{y^s_\emptyset}{y_\emptyset}\cdot \err^s_{vw}.
\end{align*}
as the error for $vw$ conditioned on $s$, and the unconditioned error. Notice that all these quantities are expectations, and thus deterministic.  In our analysis, we shall leave the error terms in all inequalities, and bound them in the next section. 

\begin{lemma}
    \label{lemma:KT-prob-clustered}
    Given a vertex $v \in V'$, the probability that $v$ is clustered in the execution of set-based-cstr-clst is exactly $\frac1{y_\emptyset}$.
\end{lemma}
\begin{proof}
    The probability is 
    \begin{align*}
        \sum_{s}\frac{y^s_\emptyset}{y_\emptyset}\sum_{u \in V'}\frac{y^s_u}{sy^s_\emptyset} \cdot \frac{y^s_{uv}}{y^s_u} = \frac1{y_\emptyset}\sum_{s}\frac1s\sum_{u \in V'}y^s_{uv} = \frac1{y_\emptyset}\sum_{s} y^s_v = \frac1{y_\emptyset}y_v = \frac1{y_\emptyset}. 
    \end{align*}
    The second equality is by \eqref{LPC:KT-size-s}. The third and the last inequalities are by \eqref{LPC:KT-define-s} and \eqref{LPC:KT-a-contained} respectively. 
\end{proof}

\begin{lemma}
    Focus on an edge $vw \in {V' \choose 2}$.
    \begin{enumerate}
        \item The probability that $vw$ is decided (i.e., one of $v$ and $w$ is in $C$) by the execution of set-based-cstr-clst is at least $ \frac1{y_\emptyset} (1  + \tilde x_{vw}) - \err_{vw}$.
        \item If $vw \in E^+$, then the probability that $vw$ is decided wrongly is at most $\frac{2}{y_\emptyset} \cdot \tilde x_{vw} + \err_{vw}$.
        \item If $vw \in E^-$, then the probability that $vw$ is decided wrongly is at most $\frac{1}{y_\emptyset} \cdot (1-\tilde x_{vw}) + \err_{vw}$.
    \end{enumerate}
\end{lemma}

\begin{proof}
    We focus on the first statement. 
    The probability that $vw$ is decided conditioned on $s$ is at least 
    \begin{align*}
    	&\quad \sum_{u \in V'}\frac{y^s_u}{s y^s_\emptyset}\cdot \left(\frac{1}{y^s_u} \cdot \big(y^s_{uv} + y^s_{uw} - y^s_{uvw}\big) - \err^s_{vw|u}\right)=\sum_{u \in V'}\left(\frac{1}{sy^s_\emptyset}\cdot(y^s_{uv} + y^s_{uw} - y^s_{uvw}) - \frac{y^s_u}{sy^s_{\emptyset}}\cdot\err^s_{vw|u}\right)\\
    	&=\frac{1}{y^s_\emptyset}(y^s_{v}+y^s_{w}-y^s_{vw}) - \err^s_{vw}.
    \end{align*}
    
    To see the second equality, we apply \eqref{LPC:KT-size-s} with $S = \{v\},\{w\}$ and $\{v, w\}$ respectively, and use the definition of $\err^s_{vw}$.
    
    Deconditioning on $s$, we have that the probability $vw$ is decided is at least
    \begin{align*}
    	&\quad \sum_{s}\frac{y^s_\emptyset}{y_\emptyset}\cdot \left(\frac{1}{y^s_\emptyset}(y^s_v+y^s_w -y^s_{vw}) - \err^s_{vw}\right) =  \frac{1}{y_\emptyset}\sum_{s}(y^s_v+y^s_w -y^s_{vw}) - \err_{vw}\\
    	&=  \frac1{y_\emptyset} (1  + 1 - y_{vw} ) - \err_{vw} =   \frac1{y_\emptyset} (1  + \tilde x_{vw}) - \err_{vw}.
    \end{align*}
    The first equality used the definition of $\err_{vw}$. The third equality used that $\sum_{s} y^s_v = y_v = 1, \sum_{s} y^s_w = y_w = 1$ and $\sum_{s}y^s_{vw} = y_{vw} = 1 - \tilde x_{vw}$.

    The second and third statements can be proved similarly. When $vw \in E^+$, the probability that $vw$ is wrongly decided conditioned on $s$ and $u$ is at most $\frac{1}{y^s_u}(y^s_{uv} + y^s_{uw} - y^s_{uvw}) + \err^s_{vw|u}$; when $vw \in E^-$, this is at most $\frac{y^s_{uvw}}{y^s_u} + \err^s_{vw|u}$. Following the calculations, the two unconditioned probabilities are at most $\frac1{y_\emptyset}(1 + 1 - 2y_{vw}) + \err_{vw} = \frac{1}{y_\emptyset}\cdot 2\tilde x_{vw} + \err_{vw}$ and $\frac{1}{y_\emptyset} \cdot y_{vw} + \err_{vw} = \frac{1}{y_\emptyset} \cdot (1 - \tilde x_{vw}) + \err_{vw}$.
\end{proof}
Therefore, if we ignore the error terms, and the difference between $\tilde x_{vw}$'s and $x_{vw}$'s, the algorithm exactly simulates the KT-rounding algorithm for the bounded-cluster size case.  By \eqref{LPC:KT-cost}, the difference between $\tilde x$ and $x$ is small. In the next section, we shall bound the errors. 

\subsection{Handing the Errors}  In this section, we handle the errors. Again recall that we focus on one execution of the set-based-cstr-clst procedure, with input $V'$ and vector $y$, which defines the vector $\tilde x$. The key lemma we prove is the following:
\begin{lemma}
    \label{lemma:KT-bound-err}
	\begin{align}
		\sum_{vw \in {V' \choose 2}} \err_{vw} \leq \eps \cdot \E\left[\Big|\{uw \in {V' \choose 2} \cap E_\adm: uw\text{ decided}\}\Big|\right].\label{inequ:error-to-budget}
	\end{align}
\end{lemma}
\begin{proof} Through the proof, we assume $u, v, w$ are all in $V'$, $vw$ and $uw$ are in $V' \choose 2$. We bound the sum of errors conditioned on $s$:
	\begin{align*}
		\sum_{vw} \err^s_{vw} &= \sum_{vw} \frac{1}{s y^s_\emptyset}\cdot \sum_{u \in V'}y^s_u\cdot\err^s_{vw|u} =  \frac1{sy^s_\emptyset} \sum_{u \in V'}y^s_u\cdot\sum_{vw}\err^s_{vw|u}.
	\end{align*}	
	Fix some $s \in [n'], u \in V'$, and we now bound $\sum_{vw} \err^s_{vw|u}$.  If $s = |K_u|$, then $C$ will be $K_u$ and no errors will be created. (Notice that in this case the LP constraints will imply that $y^s_{uv} = 0$ for every $v \notin K_u$.) So, we assume $s > |K_u|$. By \eqref{LPC:KT-forbid}, we have that $s > |K_u| + \eps\cdot d_\adm(u)$, since otherwise we shall $y^s_u = 0$. 
 Finally, by Lemma~\ref{lem:RT-KT}, $\sum_{vw} \err^s_{vw|u} \leq \eps_r |N_\adm(u) \cap V'|^2$, and recall that $\eps_r = \eps^2$.
 Therefore, 
 \begin{align*}
		\sum_{vw} \err^s_{vw|u} &\leq \eps_r \cdot |N_\adm(u) \cap V'|^2 \leq  \frac{\eps_r}{\eps} \cdot |N_\adm(u) \cap V'|  \cdot (s - |K_u|) \\
		&= \eps \cdot |N_\adm(u) \cap V'| \cdot \sum_{v \in N_\adm(u)\cap V'}\frac{y^s_{uv}}{y^s_u}\\
        &= \eps \cdot \sum_{v, w \in N_\adm(u) \cap V'} \frac{y^s_{uv}}{y^s_u}.
	\end{align*}

	So we have
	\begin{align*}
		\sum_{vw}\err^s_{vw} \leq \frac{1}{sy^s_\emptyset} \sum_{u \in V'} y^s_u \cdot \sum_{v, w \in N_\adm(u) \cap V'} \eps \cdot \frac{y^s_{uv}}{y^s_u}  \leq \eps \cdot \frac{1}{sy^s_\emptyset}\cdot \sum_{u \in V', v, w \in N_\adm(u) \cap V'} y^s_{uv}.
	\end{align*}
	
	Now we consider the right side of \eqref{inequ:error-to-budget}. The expectation of the quantity conditioned on $s$ is at least 
	\begin{align*}
		\eps\cdot\sum_{v \in V'} \frac{y^s_v}{sy^s_\emptyset} \sum_{u \in N_\adm(v)\cap V', w \in N_\adm(u)} \frac{y^s_{uv}}{y^s_v}= \eps  \cdot\frac{1}{sy^s_\emptyset}\cdot \sum_{u \in V', v, w \in N_\adm(u) \cap V'} y^s_{uv}.
	\end{align*}
    This is at least $\sum_{vw} \err^s_{vw}$. Taking all $s$ into consideration gives us \eqref{inequ:error-to-budget}. 
\end{proof}

\paragraph{Wrapping Up} Now we finish the proof of Theorem~\ref{thm:KT}. Focus on one execution of set-based-cstr-clst. We define the following three types of \emph{budgets}:
\begin{itemize}
	\item When an edge $vw \in {V' \choose 2}$ is decided in the procedure, we get an \emph{LP budget} of $\frac{2x_{vw}}{1 + x_{vw}}$ from $vw$ if $vw$ is a $+$edge, and $\frac{1 - x_{vw}}{1 + x_{vw}}$ if it is a $-$edge. This is the budget coming from the cost of $x$. 
	\item If additionally $vw$ is an admissible edge, we get an \emph{error budget} of $\eps$ from $vw$. This will be used to cover the errors incurred by the RT rounding procedure. 
	\item Finally a vertex $v$ is clustered in the procedure, we get a \emph{difference budget} of $2\eps \cdot d_\adm(v)$. This will be used to cover the difference between $\tilde x$ and $x$. 
\end{itemize}

If an edge $vw \in {V' \choose 2}$ is wrongly decided, we pay a cost of $1$. For a $+$edge $vw \in {V' \choose 2}$,
\begin{align*}
	\Pr[vw \text{ wrongly decided}] &\leq \frac1{y_\emptyset} \cdot 2\tilde x_{vw} + \err_{vw} = \frac{1}{y_\emptyset} \cdot \left( 2 x_{vw} + 2(\tilde x_{vw} - x_{vw}) \right) + \err_{vw}\\
	&= \E[\text{LP budget from }vw] + \frac{2(\tilde x_{vw} - x_{vw})}{y_\emptyset} + \err_{vw}.
\end{align*}
Similarly, for a $-$edge $vw \in {V' \choose 2}$, we have 
\begin{align*}
	\Pr[vw \text{ wrongly decided}] &\leq  \E[\text{LP budget from }vw] + \err_{vw}.
\end{align*}

Summing up the inequalities over all edges $vw \in {V' \choose 2}$, we have that in the execution of set-based-cstr-clst, 
\begin{align*}
	&\quad \E[\text{cost incurred}]  \leq \E\left[\text{LP budget from  }{V' \choose 2}\right] + \frac{2}{y_\emptyset}\sum_{vw}(\tilde x_{vw} - x_{vw}) + \sum_{vw} \err_{vw}\\
	&\leq \E\left[\text{LP budget from  }{V' \choose 2}\right] + \frac{2}{y_\emptyset}\cdot \eps \sum_{u \in V'} d_\adm(u) + \E\left[\text{error budget from }{V' \choose 2}\cap E_{\adm}\right]\\
	&= \E\left[\text{all 3 types of budget we get}\right].
\end{align*}
The second inequality is by \eqref{LPC:KT-cost}, and Lemma~\ref{lemma:KT-bound-err}. To see the last equality, notice that every $u$ is clustered with probability exactly $\frac{1}{y_\emptyset}$ by Lemma~\ref{lemma:KT-prob-clustered}. 

The procedure can be derandomized by enumerating $s, u$ and the random seeds used in Raghavendra-Tan rounding procedure.\footnote{We remark that the derandomization step is necessary due to our round-or-cut paradigm. Since we can not control weather the future iterations will succeed or fail (i.e., return a separation plane), we could not guarantee that the expected cost we pay is the expected budget we get, conditioned on that the algorithm succeeds.} So, we can guarantee that in every execution of set-based-cstr-clst, the cost incurred is at most the budget we get, including the LP, error and difference budgets.  So, considering the whole Algorithm~\ref{alg:KT-rounding}, the cost of the clustering is at most the total budget we get.  For every edge $vw \in {V \choose 2}$, we only get the LP and error budget once from $vw$. For every vertex $v \in V$, we only get the error budget once from $v$. Therefore, the total budget we get is 
\begin{align*}
    \sum_{vw \in E^+} \frac{2x_{vw}}{1 + x_{vw}} + \sum_{vw \in E^-} \frac{1 - x_{vw}}{1 + x_{vw}} + \eps |E_{\adm}| + 2\eps \sum_{u \in V}d_{\adm}(u) \leq \sum_{vw \in E^+} \frac{2x_{vw}}{1 + x_{vw}} + \sum_{vw \in E^-} \frac{1 - x_{vw}}{1 + x_{vw}} + O(\eps) \cdot |E_{\adm}|.
\end{align*}
	
We rescale $\eps$ so that the additive term becomes $\eps |E_\adm|$. This finishes the proof of Theorem~\ref{thm:KT}.

    \section{Pivot-Based Rounding Procedure}
\label{sec:CLN}
In this section, we present our pivot-based rounding algorithm $\calA_{\mathrm{pivot}}$ and prove its guarantee restated below. 

\cln*

In this section, a singleton vertex in $V \setminus V_{\calK}$ will be treated almost like an atom, so the term atom will also refer to (the singleton set of) such a vertex. 
Let $\calK' := \calK \cup (\cup_{v \in V \setminus V_{\calK}} \{\{ v\}\})$ be the set of atoms.
We will also assume that two vertices in the same atom have exactly the same neighbors in $E_{\adm}$. This can be ensured by, for every triple $uvw$ with $u, v$ in the same atom and $vw \in E_{\adm}$, $uw \notin E_{\adm}$, removing $vw$ from $E_{\adm}$; any good clustering cannot put $v$ and $w$ in the same cluster, so it is safe to remove it. 

\subsection{Relaxation and Algorithm}
Given $x \in [0, 1]^{\binom{V}{2}}$, for some constant $r = O(1/\eps^2)$, we use the following LP.
For a set $S \subseteq V$, $y_S$ indicates whether $S$ is a subset of a cluster in an optimal solution. 
We also add the constraint that $x_{uv} = 1- y_{uv}$, which will ensure that $y_{uv} = 1$ if $u$ and $v$ are in the same atom and $y_{uv} = 0$ if $uv$ is non-admissible.

Apart from the constraints imposed by $x$, this is a weaker version of $r$-rounds of the Sherali-Adams relaxation introduced by~\cite{CLN22}.
We present this relaxation to present a smaller set of constraints used in our algorithm (e.g., this relaxation also works with the rounding algorithm of~\cite{CLN22}) and be consistent with the LP introduced in Section~\ref{sec:KT}. 
In particular, apart from the constraint~\eqref{eq:CLN-Triangle} used in the triangle analysis (its LHS indicates the event all $u, v, w$ are in different clusters; see Section~\ref{sec:triangle-ideal}), it is a simpler version of the LP used in Section~\ref{sec:KT} that does not distinguish sets of different sizes and does not involve additional variables $\tilde{x}$.

\begin{align}
  &y_{uv} = 1-x_{uv} && \forall uv \in \binom{V}{2} \\
  &y_{u} = 1  && \forall u \in V \\
  &\sum_{T' \subseteq T}(-1)^{|T'|}y_{S\cup T'} \in [0, y_S] && \forall S\cap T =\emptyset, |S \cup T| \leq r 
  \label{eq:CLN-SA} \\
  & 1 - \big( 
  y_{vw} + 
  y_{wu} +
  y_{uv} - 2y_{uvw} \big)
  \geq 0
  && \forall u, v, w \in V
  \label{eq:CLN-Triangle} \\
  &y  \geq 0. \label{nonneg_y}
\end{align}
Any integral good clustering with respect to $(\calK, E_{\adm})$ yields a feasible solution to the above LP. Therefore, if the above LP is not feasible for the given $x$, it
implies that $x$ is not a convex combination of good clustering, and 
by Lemma~\ref{lemma:find-cut}, the algorithm can yields a hyperplane separating $x$ from the convex hull of good clusterings. 

For the rest of the section, we assume that $y$ is an optimal solution for the above relaxation in the initial graph. The rounding algorithm is given in Algorithm~\ref{alg:CLN}. 
Let $N^+(v)$ and $N^+_{\adm}(v)$ denote the neighbors of $v$ with respect to $E^+$ and $E^+ \cap E_{\adm}$ respectively. $N^-(v)$ and $N^-_{\adm}(v)$ are defined similarly.

\begin{algorithm}[h!]
\caption{Pivot-based Rounding
\label{alg:CLN}}
\begin{algorithmic}[1]
		\Require{A Correlation Clustering instance $G = (V, E^+ \cup E^-)$, a preclustered instance $(\calK, E_{\adm})$ for $G$, and a metric $x \in [0, 1]^{V \choose 2}$ satisfying properties in Theorem~\ref{thm:round-or-cut}}
		\Ensure{Either a separation plane for $x$, or an integral clustering $\calC$}
\State Run the above LP to obtain a feasible solution $y$. If not feasible, return a separation plane for $x$ 
\State $\calC \leftarrow \emptyset$ 
\While{$V \neq \emptyset$}
\State $S \leftarrow \mathbf{Cleanup}(V)$ (Algorithm~\ref{alg:cleanup-CLN})
\If{$S \neq \emptyset$}
\State $\calC \leftarrow \calC \cup \{ S \}$, $V \leftarrow V \setminus S$, {\bf continue}
\EndIf
\State Pick a pivot $p \in V$ uniformly at random
\For{$v \in N^-(p)$}
\State Add $v$ independently to $S^{-}$ with probability $y_{pv}$ 
\EndFor
\State Let $K$ be the atom containing $p$
\State\label{step:RT-CLN} Sample $S^+ \subseteq N^+_{\adm}(p)$ from the correlated rounding procedure (Lemma~\ref{lem:RT})
\State $S \leftarrow S^- \cup S^+ \cup K$, $\calC \leftarrow \calC \cup \{ S \}$, $V \leftarrow V \setminus S$ 
\EndWhile
\State {\bf return} $\calC$
\end{algorithmic}
\end{algorithm}

The cleanup subroutine is described in Algorithm~\ref{alg:cleanup-CLN}. Let $E^+(V, K)$ be the set of $+$edges incident on $K$.
Its meaning will be more clear after we introduce the setup for the analysis using budgets in Section~\ref{sec:CLN-setup}.

\begin{algorithm}[h!]
\caption{{\bf Cleanup}($V$)~\label{alg:cleanup-CLN}}
\begin{algorithmic}[1]
\For{each atom $K \in \calK'$}
\State Let $ALG_K$ be the cost of removing $K$ as a single cluster, which is the number of $-$edges inside $K$ plus the number of $+$edges between $K$ and $V \setminus K$. 
\State Let $\Delta_K$ be the decrease in the budget if we remove $K$ from $V$; formally, $\Delta' := (\sum_{uv \in E^+(K, V)} \min(1.515+x_{uv}, 2)x_{uv}) + 
(\sum_{uv \in E^-(K, V)} 2(1-x_{uv}))
+
\eps \cdot |\{ uv \in E_{\adm} : u \in K \mbox{ or } v \in K \}|.
$ 
\State If $\Delta_K \geq ALG_K$, {\bf return} $K$
\EndFor
\State {\bf return} $\emptyset$ 
\end{algorithmic}
\end{algorithm}

Algorithm~\ref{alg:CLN} also uses the correlated rounding procedure of Raghavendra and Tan~\cite{RT12}, adapted for \cc in~\cite{CLN22}. The guarantee of the correlated rounding is as follows.
\begin{lemma}
\label{lem:RTrestated}
  In Step~\ref{step:RT-CLN} of Algorithm~\ref{alg:CLN}, one can sample $S^+ \subseteq N^+_{\adm}(p)$ in time $n^{O(r)}$ such that
  \begin{itemize}
  \item For each $v \in N^+_{\adm}(p)$, $\Pr[v \in S^+] = y_{pv}$. 
  \item $\E_{u, v \in N_{\adm}^+(p)}[|\Pr[u, v \in S^+] - y_{puv}|] \leq \epsrt$, where $\epsrt = O(1/\sqrt{r})$.
  \end{itemize}
\label{lem:RT}
\end{lemma}

Even though we present the algorithm as a randomized algorithm, a standard derandomization using conditional expectations will deterministically yield a clustering whose cost is as good as the expected value~\cite{CLN22}. 

\subsection{Setup for Analysis}
\label{sec:CLN-setup}

Our high-level setup of the analysis follows from~\cite{CLN22}, which in turn is based on~\cite{ACN08} and~\cite{CMSY15}, with a slight change that instead of the total LP value remaining in the current instance, we consider a more general {\em budget} remaining in the current instance, which already incorporates the approximation ratio. 
For each $uv \in E^+$, give the budget of $\min(1.515 + x_{uv}, 2)x_{uv}$, and 
for each $uv \in E^-$, give the budget of $2(1 - x_{uv})$ (call them the {\em LP budget}). Furthermore, for each $uv \in E_{\adm}$, give an additional budget of $\eps$ (call it the {\em error budget}).
Then it is clear that the total initial budget is exactly the desired upper bound on the (expected) cost of the clustering (the RHS of Theorem~\ref{thm:CLN}). 

Consider the $t$-th iteration of
Algorithm~\ref{alg:CLN} with the vertex set $V_t$ and suppose that the algorithm obtains the cluster $S$ in this iteration by building it from a random pivot $p$ (instead of using the cleanup). 
Let $\costr_p(u, v)$ be the probability that $uv$ is violated in the
rounding algorithm when $p$ is the pivot, and $\lpr_p(u, v)$ be the total budget of $uv$ (LP and error budget combined) times the probability that $uv$ disappears (i.e.,
$\Pr[S \cap \{ u, v \} \neq \emptyset]$).  The superscript $r$ stands for (actual) rounding.

We call a set of three
distinct vertices a {\em triangle}. A set of two vertices is called a
{\em degenerate triangle}.  For triangle $uvw$, let
$\costr(u, v, w) = \costr_u(v, w) + \costr_v(u, w) + \costr_w(u, v)$
and $\lpr(u, v, w) = \lpr_u(v, w) + \lpr_v(u, w) + \lpr_w(u, v)$. For
degenerate triangle $uv$, let $\costr(u, v) = \costr_u(u, v) +
\costr_v(u, v)$ and $\lpr(u, v) = \lpr_u(u, v) + \lpr_v(u, v)$. Let
\[
ALG_t := \E_{u \in V} \sum_{vw \in \binom{V_t}{2}} \costr_u(v, w)
\]
be the expected cost incurred by this iteration, and 
\[
\Delta_t:= \E_{u \in V} \sum_{vw \in \binom{V_t}{2}} \lpr_u(v, w)
\]
be the expected amount of the budget removed by this iteration.
If we could show that $ALG_t \leq \Delta_t$ for all $t$, 
we will get an upper bound on the total cost $\mathbf{ALG}$ as
\[
\E[\mathbf{ALG}] = \E[\sum_{t=0}^R ALG_t] \leq \E[\sum_{t=0}^R \Delta_t] = \Phi
\]
where $\Phi$ is the initial budget and $R$ is the number of the iterations.

Notice that, even when the cluster $S$ is chosen from the cleanup step, $ALG_t$ and $\Delta_t$ can be still defined as the incurred cost and removed budget respectively, and the design of the cleanup step (Algorithm~\ref{alg:cleanup-CLN}) exactly ensures that we remove $S$ when $ALG_t \leq \Delta_t$ deterministically.

Therefore, in order to prove Theorem~\ref{thm:CLN}, it suffices to consider one iteration where $S$ is built from a random pivot $p$. For the rest of Section~\ref{sec:CLN}, let us omit the subscript $t$ denoting the iteration. We prove $ALG \leq \Delta$, which is equivalent to showing 
\begin{align}
\sum_{uvw \in \binom{V}{3}} \costr(u,v, w) + \sum_{uv \in \binom{V}{2}} 
\costr(u,v)
\leq 
\sum_{uvw \in \binom{V}{3}} \lpr(u,v, w) + \sum_{uv \in \binom{V}{2}} \lpr(u,v).
\label{eq:triangle-sum}
\end{align}

Recall that a triangle is $+++$ if it has three $+$edges and $++-$,
$+--$, $---$ triangles are defined similarly.  For a degenerate
triangle $uv$, $\costr_u(u, v)$ and $\lpr_u(u, v)$ depend only
on $x_{uv}$ and the sign of $uv$.  Even for a triangle $uvw$, the values of $\costr_u(v, w)$ and $\lpr_u(v, w)$ only depend on
$x_{uv}, x_{uw}, x_{vw}$ and the signs and admissibilities of the edges unless both $uv$ and $uw$ are admissible $+$edges; $v$ and $w$ are added to $S$ independently with the probabilities depending on $x_{uv}$ and
$x_{uw}$ respectively. 
When both $uv$ and $uw$ are admissible $+$edges, then they are rounded using Lemma~\ref{lem:RT}, and $\Pr[v, w \in S^+
  | u\mbox{ is pivot}]$ must be, ideally, exactly equal to $y_{uvw}$,
but Lemma~\ref{lem:RT} only gives an approximate guarantee
amortized over the vertices in $N^+_{\adm}(u)$.

We define the following idealized quantities $\costi(\cdot)$ and $\lpi(\cdot)$. 
(The superscript $i$ stands for ideal.)
Intuitively, $\costi(\cdot)$ and $\lpi(\cdot)$ are defined assuming that the correlated rounding for admissible $+$edges are perfect, and we do not consider the error budget for $\Delta$. Formally, $\costi_u(\cdot), \costi(\cdot), \lpi_u(\cdot), \lpi(\cdot)$ are defined identically to $\costi_u(\cdot)$, $\costi(\cdot)$, $\lpi_u(\cdot)$, $\lpi(\cdot)$ respectively, assuming that in Step~\ref{step:RT-CLN} of Algorithm~\ref{alg:CLN},
the condition (2) is replaced by $\Pr[u, v \in S^+ | p\mbox{ is pivot}] = y_{puv}$ for every $p \in V$, $u, v \in N^+_{\adm}(p)$ and the error budget of $\eps$ is not accounted in $\Delta$'s. With this assumption, note that for every  triangle $uvw$ both $\costi(u, v, w)$ and $\lpi(u, v, w)$ depend only on the signs of the edges and the Sherali-Adams solution induced by $uvw$ (i.e., $y_{uvw}, y_{uv}, y_{vw}, y_{wu}$). 
Then one can show that $\costi(T) \leq \lpi(T)$ for any triangle $T$. The proof of the following lemma appears in Section~\ref{sec:triangle-ideal}.

\begin{lemma}
For any triangle $T$, $\costi(T) \leq \lpi(T)$.
\label{lem:triangle_ideal}
\end{lemma}

    \subsection{Incorporating Errors}

This subsection shows how to incorporate errors and finishes the proof of Theorem~\ref{thm:CLN} assuming Lemma~\ref{lem:triangle_ideal}.  
To prove the theorem, as explained in Section~\ref{sec:CLN-setup}, it suffices to show that $\E[ALG] \leq \E[\Delta]$ in one iteration 
where {\bf Cleanup}$(V)$ returns $\emptyset$ and Algorithm~\ref{alg:CLN} proceeds by choosing a random pivot $p$.  

When $p$ is chosen as a pivot, let $N^{+}_{\adm}(p)$ be the set of vertices connected to $p$ via an admissible $+$edge. By Lemma~\ref{lem:RT}, compared to the ideal case, $ALG$ is increased by at most $\epsrt \cdot \binom{|N^+_{\adm}(p)|}{2}$. The main challenge is to show that this can be compensated by the increase in $\Delta$ due to the additional error budget $\eps$ for each admissible edge.

Let 
$\err_{vw|u}$
be $\epsrt$ when both $uv$ and $uw$ are admissible $+$ edges and $0$ otherwise. Recall that $\epsrt = O(1/\sqrt{r})$ be the error parameter from Lemma~\ref{lem:RT}. 
Similarly, $\Delta'_u(v, w) := (1 - \min(x_{uv}, x_{uw})) \eps$ is an {\em lower bound} on the expected error budget decrease from if $vw$ if it is admissible, and $0$ otherwise. 
 (It is an lower bound because when $u$ is the pivot, since our rounding algorithm satisfies $\Pr[v \in S] = y_{uv}$ for all $v$, $vw$ will be removed with probability at least $1 - \min(x_{uv}, x_{uw})$.) By letting
\[
ALG^i := \sum_{u \in V, vw \in \binom{V}{2}} \costi_u(v, w)
\quad \mbox{ and } \quad 
ALG' := \sum_{u \in V, vw \in \binom{V}{2}} \err_{vw|u} 
\]
and 
\[
\Delta^i := \sum_{u \in V, vw \in \binom{V}{2}}  \lpi_u(v, w) 
\quad \mbox{ and } \quad 
\Delta' := \sum_{u \in V, vw \in \binom{V}{2}} \Delta'_u(v, w).
\]
we have that 
\[
|V| \cdot ALG \leq ALG^i + ALG'
\]
and 
\[
|V| \cdot \Delta \geq \Delta^i + \Delta'.
\]
Since Lemma~\ref{lem:triangle_ideal} shows that $ALG^i \leq \Delta^i$, it suffices to show that $ALG' \leq \Delta'$. 

We prove it on an atom-by-atom basis. Fix an atom $K \in \calK'$ (again, it might be a singleton outside $V_{\calK}$). Let $N$ be the set of their admissible neighbors. Recall that if every vertex in $K$ has exactly the same set of neighbors with respect to $E_{\adm}$.
Let $E' = K \times N$ be the set of pairs between them (all admissible), which is partitioned into $E'_+$ and $E'_-$. 
Then $K$'s contribution to $ALG'$ is 
\[
ALG'_K := 
\sum_{u \in K, vw \in \binom{N}{2} : uv, uw \in E'_+}   \epsrt \leq |K||N|^2\epsrt,
\]
such that $\sum_{K \in \calK'} ALG'_K = ALG'$. We can also define $K$'s contribution to $\Delta'$ as
\[
\Delta'_K := \sum_{u \in N, vw \in E'} \Delta'_u(v, w) = 
\sum_{u \in N, vw \in E'} 
(1 - \min(x_{uv}, x_{uw})) \eps,
\]
so that $\sum_{K \in \calK'} \Delta'_K \leq 2\Delta'$. (For each triple $uvw$, $\Delta'_u(v, w)$ is counted at most twice, for atoms containing $v$ or $w$.) Therefore, it suffices to argue that $2ALG'_K \leq \Delta'_K$. 

We use the fact that $K$ was not removed as a single cluster from the cleanup stage, which means that 
\begin{align*}
& ALG_K = 
|E^- \cap \binom{K}{2}| 
+ |E^+ \cap (K \times (V \setminus K))| 
\\
\geq \quad & \Delta_K = 
 \sum_{uv \in (E^+ \cap (K \times V))} \min\{1.515 + x_{uv}, 2 \} \cdot x_{uv} + \sum_{uv \in (E^- \cap (K \times V))} 2(1 - x_{uv}) + \eps |E'|.
\end{align*}
Note that for each $e \in E^- \cap \binom{K}{2}$, it is an atomic edge with $x_e = 0$, so it contributes exactly $1$ to $ALG_K$ and at least $1$ to $\Delta_K$. Similarly, for each $e \in (E^+ \cap (K \times V)) \setminus E'_+$, it is a non-admissible edge with $x_e = 1$, so it contributes exactly $1$ to $ALG_K$ and at least $1$ to $\Delta_K$. Then, by only considering edges in $E'$, one can conclude that 
\[
|E'_+|
\geq \quad 
 \sum_{uv \in E'_+} \min\{1.515 + x_{uv}, 2 \} \cdot x_{uv} + \sum_{uv \in E'_-} 2(1 - x_{uv}) + \eps |E'|,
\]
which implies that $|E'_+| \geq \eps |E'|$ and $|E'_+| \geq \sum_{uv \in E'_+} 1.5 x_{uv}$, so that $\E_{uv \in E'_+} [x_{uv}] \leq 2/3$ and $\E_{uv \in E'} [x_{uv}] \leq 1-\eps/3$. 
Then one can lower bound $\Delta'_K$ as 
\[
\Delta'_K = 
 \sum_{u \in N, vw \in E'} 
\eps(1 - \min(x_{uv}, x_{uw}))
\geq \sum_{w \in N, u \in N, v \in K} 
\eps(1 - x_{uv})
= |N| \sum_{uv \in E'} \eps(1 - x_{uv}) \geq |N|^2 |K| (\eps/3).
\]
Therefore, by ensuring that $\epsrt \leq \eps/6$, we can ensure that $ALG'_K \leq \Delta'_K/2$ for every $K$ and $ALG' \leq K'$ eventually.

    \subsection{Analysis of Error-Free Version}
\label{sec:triangle-ideal}
This subsection proves Lemma~\ref{lem:triangle_ideal}. 
Specifically, for every triangle $abc$, we prove
    \begin{align}
		&\quad \cost^i_a(b,c)+\cost^i_b(c, a)+\cost^i_c(a, b)
		\leq 
  \Delta^i_a(b,c)+
  \Delta^i_b(c,a)+ \Delta^i_c(a,b).
  \label{inequ:cost-to-budget}
	\end{align}
	When two of $a, b, c$ are the same vertex, it is easy to see that~\eqref{inequ:cost-to-budget} is true~\cite{CLN22}. Therefore, we assume that $a, b, c$ are distinct vertices. We consider the four types of triangles according to the signs of their edges.
 Recall the coefficient of the LP budget for a \pedge $e$ is $f(x_{e})$ defined as $f(x) = \min(1.515+x, 2)$, whereas the coefficient for a \medge $e$ is always 2. 
 Also, note that given $y_{ab}, y_{bc}, y_{ca}, y_{abc}$, we can define $y_{a|b} = y_a - y_{ab}$ indicating the event $a$ and $b$ are separate (similarly for $y_{b|c}, y_{c|a}$), and $y_{a|bc} = y_{bc} - y_{abc}$ indicating the event $b$, $c$ are together but $a$ is separate, and $y_{a|b|c} = 1 - (y_{ab} + y_{bc} + y_{ca} - 2y_{abc})$ indicating the event $a$, $b$, $c$ are all separate. By the constraint~\eqref{eq:CLN-Triangle}, all of them are nonnegative with $y_{abc}+y_{a|bc}+y_{b|ca}+y_{c|ab}+y_{a|b|c} = 1$, so they exactly define the probability distribution over partitions on $\{ a, b, c\}$. (These variables were explicitly defined in the Sherali-Adams relaxation of~\cite{CLN22}.) 

	\paragraph{$+++$ triangles} \cite{CLN22} already proved that in their setting, the approximation ratio for $+++$ triangles is at most $1.5$. In our setting, as we have $f(x) \geq 1.5$ for every $x\in [0, 1]$, \eqref{inequ:cost-to-budget} holds.
	
	\paragraph{$---$ triangles} For a $---$ triangle $abc$, \eqref{inequ:cost-to-budget} holds even if we set the coefficients for $-$edges to 1 (instead of 2):
     \begin{align*}
         \text{left side of \eqref{inequ:cost-to-budget}} &= y_{ab}y_{ac} + y_{ab}y_{bc} + y_{ac}y_{bc}.\\ 
         \text{right side of \eqref{inequ:cost-to-budget}} &\geq (y_{ab}+y_{ac} - y_{ab}y_{ac})y_{bc} + (y_{ab} + y_{bc} - y_{ab}y_{bc})y_{ac} + (y_{ac} + y_{bc} - y_{ac}y_{bc})y_{ab} \\ 
         &= 2(y_{ab}y_{ac} + y_{ab}y_{bc} + y_{ac}y_{bc}) - 3y_{ab}y_{ac}y_{bc} \geq y_{ab}y_{ac} + y_{ab}y_{bc} + y_{ac}y_{bc}.
     \end{align*}
	
	\paragraph{$+--$ triangles} For a $+--$ triangle $abc$, we assume $ab$ and $ac$ are $-$edges and $bc$ is the $+$edge. We show that \eqref{inequ:cost-to-budget} holds even if we set the coefficient for $bc$ to $1$:
	\begin{align*}
		\text{left side of \eqref{inequ:cost-to-budget}} &= y_{ab} + y_{ac} - 2y_{ab}y_{ac} + y_{ab}y_{bc} + y_{ac}y_{bc} = (2 - x_{bc})(y_{ab} + y_{ac}) - 2y_{ab}y_{ac}. \\
		\text{right side of \eqref{inequ:cost-to-budget}} &\geq (y_{ab} + y_{ac} - y_{ab}y_{ac})x_{bc} + 2(y_{ab} + y_{bc} - y_{ab}y_{bc}) y_{ac} + 2(y_{ac} + y_{bc} - y_{ac}y_{bc}) y_{ab}\\
        &=(y_{ab} + y_{ac} - y_{ab}y_{ac})x_{bc}  + 2 (1 - x_{bc}+ y_{ab}x_{bc} )y_{ac} + 2(1 - x_{bc} + y_{ac}x_{bc})y_{ab}\\
		&= (2-x_{bc})(y_{ab} + y_{ac}) + y_{ab}y_{ac}x_{bc}.
	\end{align*}
	So the left side is at most the right side.

	\paragraph{$++-$ triangles} 
 It remains to consider the $++-$ triangle case, which contributes to the bulk of the analysis.  Focus on a $++-$ triangle $abc$. We assume $ab$ and $ac$ are $+$edges, and $bc$ is the $-$edge.  The left side of \eqref{inequ:cost-to-budget} is 
	\begin{align*}
		y_{abc} + y_{ab} + y_{bc} - 2y_{ab}y_{bc} + y_{ac} + y_{bc} - 2y_{ac}y_{bc} = y_{abc} + y_{ab} + y_{ac} + 2y_{bc} - 2(y_{ab} + y_{ac})y_{bc}.
	\end{align*}

	The right side of \eqref{inequ:cost-to-budget} is 
	\begin{align*}
		&\quad2(y_{ab} + y_{ac} - y_{abc})y_{bc} + f(x_{ac})(y_{ab} + y_{bc} - y_{ab}y_{bc})(1-y_{ac})  + f(x_{ab})(y_{ac} + y_{bc} - y_{ac}y_{bc})(1-y_{ab})\\
		&=f(x_{ac}) y_{ab} + f(x_{ab})y_{ac} + (f(x_{ab})+f(x_{a,c}))y_{bc} + (2 - f(x_{ab}) - f(x_{ac}))(y_{ab} + y_{ac})y_{bc}\\
		&\quad  - (f(x_{ab}) + f(x_{ac}))y_{ab}y_{ac} (1-y_{bc}) - 2 y_{abc}y_{bc} .
	\end{align*}
	
	The right side of \eqref{inequ:cost-to-budget} minus the left side is 
	\begin{align}
		&\quad (f(x_{ac}) - 1) y_{ab} + (f(x_{ab})-1)y_{ac} + (f(x_{ab})+f(x_{a,c}) - 2)y_{bc} + (4 - f(x_{ab}) - f(x_{ac}))(y_{ab} + y_{ac})y_{bc}\nonumber\\
		&\quad  - (f(x_{ab}) + f(x_{ac}))y_{ab}y_{ac} (1-y_{bc}) - 2 y_{abc}y_{bc} - y_{abc}.   \label{equ:diff}
	\end{align}
	We need to prove that \eqref{equ:diff} is non-negative.  We define $y_{ab|c} = y_{ab} - y_{abc}$ to indicate if $ab$ are in the same cluster that does not contain $c$; define $y_{a|bc}$ and $y_{b|ac}$ similarly. Define $y_{a|b|c} = 1 - (y_{ab} + y_{ac}+y_{bc}) + 2y_{abc}$ to indicate if $a, b, c$ are in three different clusters; this is at least $0$ by \eqref{eq:CLN-Triangle}. So $y_{ab|c}, y_{ac|b}, y_{a|bc}, y_{abc}$ and $y_{a|b|c}$ indicate the 5 cases for the clustering status of $a, b, c$. They are non-negative reals summing up to $1$. Moreover, $y_{ab} = y_{ab|c} + y_{abc}$, $y_{ac} = y_{ac|b} + y_{abc}$ and $y_{bc} = y_{a|bc} + y_{abc}$. 

    First we show that moving mass from $y_{a|bc}$ to $y_{a|b|c}$ can only decrease \eqref{equ:diff}.  Notice that this operation does not change $y_{abc}, y_{ab}, x_{ab}, y_{ac}$ and $x_{ac}$, and it decreases $y_{bc}$.   The derivative of \eqref{equ:diff} w.r.t $y_{bc}$ is \emph{at least} $f(x_{ab}) + f(x_{ac}) - 2 + (f(x_{ab}) + f(x_{ac}))y_{ab}y_{ac} - 2y_{abc}  \geq 1 + 3y^2_{abc} - 2y_{abc} \geq 0$.  Therefore, we can assume $y_{a|bc} = 0$, and thus $y_{bc} = y_{abc} \leq \min\{y_{ab}, y_{ac}\}$.  \eqref{equ:diff} becomes 
	\begin{align}
		&\quad (f(x_{ac}) - 1) y_{ab} + (f(x_{ab})-1)y_{ac} + (f(x_{ab})+f(x_{a,c}) - 3)y_{bc} + (4 - f(x_{ab}) - f(x_{ac}))(y_{ab} + y_{ac})y_{bc}\nonumber\\
		&\quad  - (f(x_{ab}) + f(x_{ac}))y_{ab}y_{ac}(1-y_{bc}) - 2y^2_{bc}. \label{equ:diff1}
	\end{align} 

	It remains to prove that \eqref{equ:diff1} is non-negative, under the condition that $y_{ab}, y_{ac},y_{bc} \in [0, 1]$ and $(y_{ab}+y_{ac}-1)_+ \leq y_{bc} \leq \min\{y_{ab}, y_{ac}\}$.   \medskip
	
	Fixing $y_{ab}$ and $y_{ac}$, \eqref{equ:diff1} is a quadratic function of $y_{bc}$. So it is minimized when $y_{bc} = (y_{ab}+y_{ac}-1)_+$ or $y_{bc} = \min\{y_{ab}, y_{ac}\}$, as the coefficient for the quadratic term $y_{bc}^2$ is $-2$.  
	We consider three cases.
	
	\paragraph{Case 1: $y_{ab} + y_{ac} \leq 1$ and $y_{bc} = 0$.} In this case, \eqref{equ:diff1} becomes 
		\begin{align}
			&\quad (f(x_{ac}) - 1) y_{ab} + (f(x_{ab})-1)y_{ac} - (f(x_{ab}) + f(x_{ac}))y_{ab}y_{ac}. \label{equ:diff2}
		\end{align}
		If both of $y_{ab}$ and $y_{ac}$ are at most $1/2$, then $x_{ab} \geq 1/2$ and $x_{ac } \geq 1/2$. Thus $f(x_{ab}) = f(x_{ac}) = 2$.  \eqref{equ:diff2} is $y_{ab} + y_{ac} - 4y_{ab}y_{ac} \geq 0$, as $y_{ab} \geq 2y_{ab}y_{ac}$ and $y_{ac} \geq 2y_{ab}y_{ac}$. 
		
		So, we assume exactly one of $y_{ab}$ and $y_{ac}$ is at least $1/2$. Wlog, we assume $y_{ab} < 1/2 \leq y_{ac}$. Thus $x_{ac}\leq 1/2 < x_{ab}$ and $f(x_{ab}) = 2$.  \eqref{equ:diff2} becomes $(f(x_{ac})-1) y_{ab} + y_{ac}-(2+f(x_{ac}))y_{ab}y_{ac}$. Fixing $y_{ac} \geq 1/2$ and $x_{ac} = 1 - y_{ac}$, the function is linear in $y_{ab}$. Thus it is minimized when $y_{ab} = 0$ or $y_{ab} = 1 - y_{ac}$. In the former case, \eqref{equ:diff2} is $y_{ac} > 0$. In the latter case, it is 
		\begin{align*}
			&\quad(f(x) - 1) x+ 1-x - (2 + f(x))x(1-x) = (x - x(1-x)) f(x) +  1 - 2x - 2x(1-x) \\
			&= x^2f(x) + 1 - 4x+2x^2, 
		\end{align*}where $x: = x_{ac}  = 1 - y_{ac} = y_{ab}\in[0, 1/2]$. 
        The number $1.515$ is chosen so that we have $\min\{1.515+x, 2\} \geq \frac{-1 + 4x - 2x^2}{x^2}$ for every $x \in [0, 1/2]$. See Figure~\ref{fig:f}. So in this case \eqref{equ:diff1} is at least $0$.
        \begin{figure}
            \centering
            \includegraphics[width=0.3\textwidth]{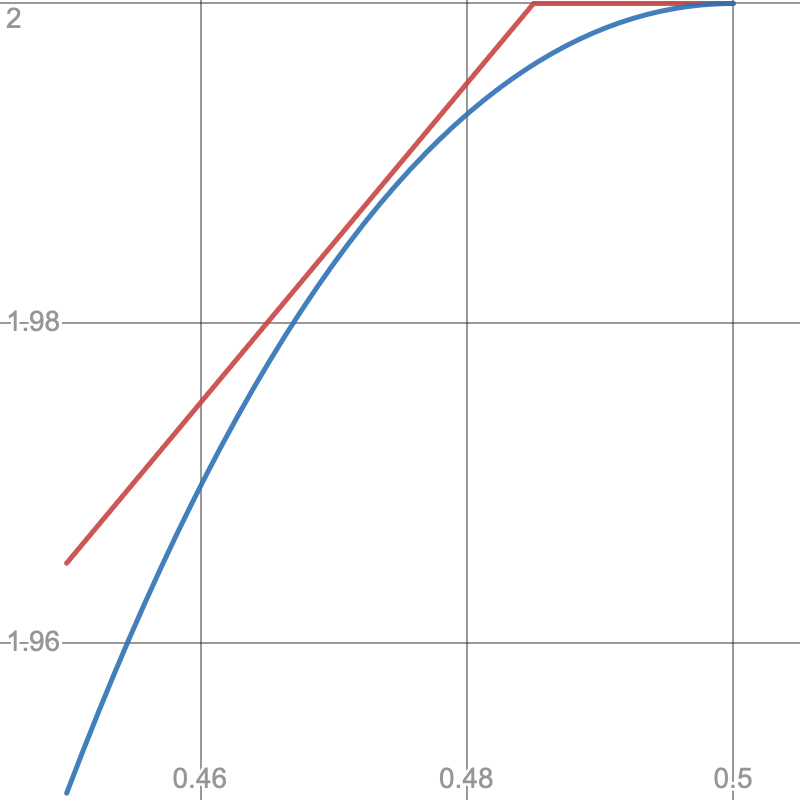}
            \caption{The function $f(x) = \min\{1.515+x, 2\}$ (the red line) and $\frac{-1 + 4x - 2x^2}{x^2}$ (the blue line) over the interval $[0.45, 0.5]$.}
            \label{fig:f}
        \end{figure}
		
		\paragraph{Case 2: $y_{ab} + y_{ac} \geq 1$ and $y_{bc} = y_{ab} + y_{ac} - 1$.} We further divide this case into many sub-cases, which by the symmetry between $y_{ab}$ and $y_{ac}$, cover the whole Case 2.
		\begin{itemize}
			\item Case 2a: $y_{ab} \leq 0.515, y_{ac} \in (0.515, 0.69]$. Then  $x_{ab}\geq 0.485$ and $x_{ac} < 0.485$.  So, $f(x_{ab}) = 2$ and $f(x_{ac}) = 1.515 + x_{ac} = 2.515 - y_{ac}$.  \eqref{equ:diff1} equals 
			\begin{align}
				&\quad (1.515 - y_{ac})y_{ab} + y_{ac} + (1.515 - y_{ac}) y_{bc}+ (y_{ac}-0.485)(y_{ab} + y_{ac})y_{bc} \nonumber\\
                &- (4.515-y_{ac})y_{ab}y_{ac} (1-y_{bc}) - 2y_{bc}^2. \label{equ:diff3}
			\end{align}
		
			We fix $y_{ac}$. Consider the operation of increasing $y_{ab}$ and $y_{bc}$ at rate 1 so as to keep $y_{bc} = y_{ab} + y_{ac} - 1$. \eqref{equ:diff3} will increase at a ratio of 
			\begin{align}
				&\quad 1.515 - y_{ac} + 1.515 - y_{ac} + (y_{ac}-0.485)(y_{ab}  + y_{ac}  + y_{bc}) - (4.515-y_{ac})y_{ac} (1-y_{bc} - y_{ab})- 4y_{bc} \nonumber\\
				&= (3.03 - 2y_{ac}) + (y_{ac} - 0.485)(2y_{ab} + 2y_{ac} -1) - (4.515-y_{ac})y_{ac} (2-2y_{ab}- y_{ac})   - 4(y_{ab} + y_{ac} - 1) \nonumber \\
                & = 7.515 - 4.97y_{ab} - 17y_{ac} + 11.03y_{ab}y_{ac} +8.515y_{ac}^2 - 2y_{ab}y_{ac}^2 - y_{ac}^3. \label{equ:deriv}
			\end{align}
            We then show \eqref{equ:deriv} is positive for $y_{ac} \in (1/2, 0.69]$.  For a fixed $y_{ac}$, the quantity is a linear function of $y_{ab}$, and the coefficient for the $y_{ab}$-term is $-4.97 + 11.03y_{ac}  - 2y_{ac}^2$. This function is at least $0$ for $y_{ac} \in [1/2, 1]$ as it is monotone increasing from $\infty$ to $11.03/4$, and its value is positive when $y_{ac} = 1/2$. Therefore \eqref{equ:deriv} is minimized when $y_{ab}$ is minimized; that is, $y_{bc} = 1 - y_{ac}$.  In this case, \eqref{equ:deriv} becomes
            \begin{align*}
                 &\quad 7.515 - 4.97(1-y_{ac}) - 17y_{ac} + 11.03(1-y_{ac})y_{ac} +8.515y_{ac}^2 - 2y_{ac}^2(1-y_{ac}) - y_{ac}^3\\
                &= 2.545 - y_{ac} - 4.515 y_{ac}^2 + y_{ac}^3.
            \end{align*}
            For $y_{ac} \in [1/2, 0.69]$, the function is at least $0$, which implies that \eqref{equ:deriv} is at least $0$. 
            
			So, the operation of decreasing $y_{ab}$ and $y_{bc}$ at the same rate can only decrease the \eqref{equ:diff3}. Thus \eqref{equ:diff3}, which is equal to \eqref{equ:diff1}, is minimized when $y_{ab} = 1 - y_{ac}$ and $y_{bc} = 0$. This is already considered in Case 1.

            \item Case 2b: $y_{ab} \leq 0.515, y_{ac} \in [0.69, 1]$. In this case, we shall simply use $1.5$ for $f(x_{a,c})$ and $2$ for $f(x_{ab})$. That is, \eqref{equ:diff1} is lower bounded by 
            \begin{align}
                &\quad 0.5 y_{ab} + y_{ac} + 0.5 y_{bc} + 0.5(y_{ab} + y_{ac})y_{bc} - 3.5 y_{ab}y_{ac}(1-y_{bc})-2y_{bc}^2. \label{equ:diff4}
            \end{align}
            Consider the operation of increasing $y_{ab}$ and decreasing $y_{ac}$ at the same rate. This does not change $y_{ab} + y_{ac}$ and $y_{bc} = y_{ab} + y_{ac} - 1$. It increases $y_{ab}y_{ac}$. Then it is easy to see that the operation will decrease the above quantity.  So, the above quantity is minimized either when $y_{ab} = 1/2$ and $y_{ac} \geq 0.69$, or when $y_{ab} < 1/2$ and $y_{ac} = 0.69$ (we shall have $y_{ab} \geq 0.31$ as we have $y_{ab} + y_{ac} \geq 1$).

            When $y_{ab} = 1/2$ and $y_{ac} \geq 0.69$, we consider the operation of increasing $y_{ac}$ and $y_{bc}$ at the same rate to maintain $y_{bc} = y_{ab} + y_{ac} - 1 = y_{ac}-1/2$. \eqref{equ:diff4} will increase at a rate of $1 + 0.5 + 0.5(y_{ab} + y_{ac} + y_{bc}) - 3.5 y_{ab}(1-y_{bc} - y_{ac}) - 4y_{bc} = 1.5 + 0.5\cdot 2y_{ac} - 1.75 \cdot (1.5 - 2y_{ac}) - 4(y_{ac} - 1/2) = 0.5y_{ac}+0.875 > 0$. So, the operation of decreasing $y_{ac}$ and $y_{bc}$ at the same rate will decrease \eqref{equ:diff4}. So, \eqref{equ:diff4} is minimized when $y_{ab} = 1/2$ and $y_{ac} = 0.69$. This will covered by the second case. 

            When $y_{ab} < 1/2$ and $y_{ac} = 0.69$, we consider the operation of increasing $y_{ab}$ and $y_{bc}$ at the same rate, so as to maintain $y_{bc} = y_{ab} + y_{ac} - 1 = y_{ab} - 0.31$. \eqref{equ:diff4} will increase at a rate of $0.5 + 0.5 + 0.5(y_{ab} + y_{ac} + y_{bc}) - 3.5 y_{ac}(1-y_{bc} - y_{ab}) - 4y_{bc} = 1 + 0.5(2y_{ab}+0.69-0.31) - 3.5 \cdot 0.69 (1.31 - 2y_{ab}) - 4(y_{ab} - 0.31) = 1.83y_{ab} - 0.73365$. This is negative when $y_{ab} < \frac{0.73365}{1.83}$ and positive when $y_{ab} > \frac{0.73365}{1.83}$. So, \eqref{equ:diff4} is minimized when $y_{ab} = \frac{0.73365}{1.83} \in [0.4, 0.401]$ and $y_{bc} = y_{ab}-0.31 \in [0.09, 0.091]$. In this case, \eqref{equ:diff4} is at least
            \begin{align*}
                0.5 \times 0.4 + 0.69 + 0.5 \cdot 0.09 + 0.5 \cdot 1.09\cdot 0.09 - 3.5\cdot 0.401\cdot 0.69 \cdot 0.91 - 2\cdot 0.091^2 \geq 0.086 > 0.
            \end{align*}
            So, we have proved that \eqref{equ:diff4} is at least $0$ in this case, which implies \eqref{equ:diff1} is at least $0$. 
   
			\item Case 2c: $y_{ab} \geq 0.515, y_{ac} \geq 0.515$. Then $x_{ab} \leq 0.485$ and $x_{ac} \leq 0.485$.  $f(x_{a,b}) \geq 1.5 + x_{ab} = 2.5 - y_{ab}$, and $f(x_{ac}) \geq 2.5 - y_{ac}$.
			\eqref{equ:diff1} is lower bounded by 
			\begin{align}
				&\quad (1.5-y_{ac})y_{ab} + (1.5-y_{ab})y_{ac} + (2-y_{ab} - y_{ac})y_{bc} + (y_{ab}+y_{ac}-1)(y_{ab} + y_{ac})y_{bc}\nonumber\\
				&-(5-y_{ab}-y_{ac})y_{ab}y_{ac}(1-y_{bc}) - 2y_{bc}^2. \label{equ:diff5}
			\end{align}
            Notice that $(1.5-y_{ac})y_{ab} +(1.5-y_{ab})y_{ac} = 1.5(y_{ab}+y_{ac})-2y_{ab}y_{ac}$. We fix the sum $y_{ab} + y_{ac}$ and $y_{bc} = y_{ab} + y_{ac} - 1$ is fixed, while changing $y_{ab}y_{ac}$. The coefficient for $y_{ab}y_{ac}$ is 
			\begin{align*}
				-2 - (5- y_{ab} - y_{ac})(1-y_{bc}) \leq 0.
			\end{align*}
			So \eqref{equ:diff5} is minimized when $y_{ab} = y_{ac}$.  Letting $y = y_{ab} = y_{ac} \in [1/2, 1]$, and $y_{bc} = 2y-1$, \eqref{equ:diff5} becomes
			\begin{align*}
				&2(1.5-y)y+(2-2y)(2y-1) + (2y-1)^2\cdot 2y - (5-2y)\cdot y^2\cdot (2-2y) - 2(2y-1)^2 \\
				&=-4y^4 + 22y^3 -32y^2 + 19y - 4.
			\end{align*}
			This is monotone over $y \in [1/2, 1]$ and so its minimum is $0$,  achieved at $y=1/2$. So in this case, \eqref{equ:diff1} is at least $0$.
		\end{itemize}
  
	\paragraph{Case 3: $y_{bc} = \min\{y_{ab}, y_{ac}\}$.} We assume $y_{bc} = y_{ab} \leq y_{ac}$ wlog.  In this case, we can use the lower bound 1.5 for both $f(x_{ab})$  and $f(x_{ac})$.  That is, \eqref{equ:diff1} is at least 
		\begin{align}
			&\quad 0.5 y_{ab}  + 0.5y_{ac} + (y_{ab} + y_{ac})y_{bc}-3y_{ab}y_{ac}(1-y_{bc}) - 2y^2_{bc} \nonumber\\
			&=  0.5 y_{ab} + 0.5y_{ac} + y_{ab}^2 + y_{ab}y_{ac} - 3y_{ab}y_{ac} + 3y_{ab}^2y_{ac} - 2y_{ab}^2\nonumber\\
			&= 0.5y_{ab} + 0.5y_{ac}-y_{ab}^2-2y_{ab}y_{ac} + 3y_{ab}^2y_{ac}. \label{equ:diff6}
		\end{align}
		Fix $y_{ab}$ in \eqref{equ:diff6}. The coefficient for $y_{ac}$ is $0.5  - 2y_{ab} + 3y_{ab}^2$, which is always non-negative.  So the quantity is minimized when $y_{ab} = y_{ac}$. Under this condition \eqref{equ:diff6} becomes 
		\begin{align*}
			y_{ab} - 3y_{ab}^2 + 3y_{ab}^3 = y_{ab} ( 1 - 3y_{ab} + 3y_{ab}^2)  \geq 0, \text{for every } y_{ab} \geq 0.
		\end{align*}
		Therefore, \eqref{equ:diff1} is non-negative in this case.
	
    So, we have proved that \eqref{inequ:cost-to-budget} holds for all triangles $abc$.

    \section*{Acknowledgements}
    Vincent Cohen-Addad, Euiwooong Lee, and Alantha Newman are grateful to Claire Mathieu and Farid Arthaud for valuable discussions at an early stage of this project.
    \bibliographystyle{alpha}
    \bibliography{references}
\end{document}